\documentclass{sig-alternate-05-2015}

\usepackage{graphicx}

\usepackage{balance}  
\usepackage[ruled, linesnumbered, vlined]{algorithm2e} 
\usepackage{subfigure}
\usepackage{amsmath}

\usepackage{multicol}
\usepackage{multirow}

\usepackage{enumitem}

\newtheorem{lemma}{Lemma}
\newtheorem{definition}{Definition}
\newtheorem{example}{Example}

\usepackage{hyperref}

\makeatletter
\def\hlinewd#1{%
\noalign{\ifnum0=`}\fi\hrule \@height #1 %
\futurelet\reserved@a\@xhline}

\def\@copyrightspace{\relax}

\makeatother

\begin{document}

\title{Computing Influence of a Product through Uncertain Reverse Skyline}

\numberofauthors{1} 
\author{
\alignauthor
Md. Saiful Islam\textsuperscript{\ddag1}, Wenny Rahayu\textsuperscript{\#2}, Chengfei Liu\textsuperscript{\dag3}, Tarique Anwar\textsuperscript{\dag4} and Bela Stantic\textsuperscript{\ddag5}\\
	\textsuperscript{\ddag}\affaddr{Griffith University, Gold Coast, Australia}\\       	  
	\textsuperscript{\#}\affaddr{La Trobe University, Melbourne, Australia}\\          
	 \textsuperscript{\dag}\affaddr{Swinburne University of Technology, Melbourne, Australia}\\       	 
	 \email{\{\textsuperscript{1}mdsaiful.islam, \textsuperscript{5}b.stantic\}@griffith.edu.au, \textsuperscript{2}w.rahayu@latrobe.edu.au, \{\textsuperscript{3}cliu, \textsuperscript{4}tanwar\}@swin.edu.au}
}

\maketitle
\begin{abstract}

Understanding the influence of a product is crucially important for making informed business decisions. This paper introduces a new type of skyline queries, called \emph{uncertain reverse skyline}, for measuring the influence of a probabilistic product in uncertain data settings. More specifically, given a dataset of probabilistic products $\mathcal{P}$ and a set of customers $\mathcal{C}$, an \emph{uncertain reverse skyline} of a probabilistic product $q$ retrieves all customers $c\in\mathcal{C}$ which include $q$ as one of their preferred products. We present efficient pruning ideas and techniques for processing the \emph{uncertain reverse skyline query} of a probabilistic product using R-Tree data index. We also present an efficient parallel approach to compute the \emph{uncertain reverse skyline} and \emph{influence score} of a probabilistic product. Our approach significantly outperforms the baseline approach derived from the existing literature. The efficiency of our approach is demonstrated by conducting extensive experiments with both real and synthetic datasets. 
\end{abstract}

\keywords{
UD-Dominance, Uncertain Reverse Skyline, Query Processing Algorithms, Parallel Computing.
}


\section{Introduction}
\label{sec:introduction}


These days we are experiencing voluminous customer preference and the product popularity rating data available from the product related websites, e.g., search queries in CarSales\footnote{http://www.carsales.com.au/}, YahooAutos\footnote{http://autos.yahoo.com/} etc and the product ratings in Amazon\footnote{https://www.amazon.com/}, eBay\footnote{http://www.ebay.com/} etc. The popularity ratings of the products in these sites can be treated as the probabilities by which the products match the customer preferences. Making intelligent use of these customer preference and popularity rating data might help production companies to optimize their (probabilistic) selling strategy or promotion plans and thereafter, increase their revenues \cite{FayX08}. To illustrate the problem settings studied in this paper, consider the datasets of wine products and the customer preferences as given in Fig. \ref{fig:exampledataset}. In general, a product is assumed to be liked by a customer if it closely matches her stated preference. However, the popularity rating of a product may also play an important role in her buying decisions in reality. For example, though $w_3$ matches the preference of the customer $c_3$ better than $w_5$, $w_5$ still has the chance to attract $c_3$ as its popularity rating is higher than $w_3$. We argue that both of the above factors need to be modeled in determining the influence of a product and discovering the favorable or popular product set for the manufacturers to sustain in the global market.


\begin{figure}[tb]
\setlength{\tabcolsep}{0.5pt}
\tiny
\centering
\begin{tabular}{cc}
\subfigure[Grape wines, $\mathcal{W}$]{
\centering
\begin{tabular}{|c|c|c||c|} \hlinewd{1 pt}
\textbf{Wines}&\textbf{1-GraCon(\%)}&\textbf{Price(\$)}&\textbf{Rating}\\ \hline
\hline
$w_1$&40&70&0.90\\ \hline
$w_2$&20&90&0.80\\ \hline
$w_3$&60&170&0.40\\ \hline
$w_4$&30&220&0.50\\ \hline
$w_5$&90&190&0.70\\ \hline
$w_6$&70&80&0.60\\ \hlinewd{1 pt}
\end{tabular}
}&
\subfigure[Customer Preferences, $\mathcal{C}$]{
\begin{tabular}{|c|c|c|} \hlinewd{1 pt}
\textbf{Preferences}&\textbf{1-GraCon(\%)}&\textbf{Price(\$)}\\ \hline
\hline
$c_1$&55&130\\ \hline
$c_2$&35&110\\ \hline
$c_3$&70&140\\ \hline
$c_4$&55&130\\ \hline
$c_5$&35&110\\ \hline
$c_6$&70&140\\ \hlinewd{1 pt}
\end{tabular}
}
\end{tabular}
\normalsize
\label{fig:exampledataset}
\vspace{-2ex}
\caption{Example datasets of (a) wines and (b) customer preferences. The ``Rating" column denotes the popularity of a wine in the market.}
\vspace{-2ex}
\end{figure}

The first operator for preference-based data retrieval over certain data is the \emph{skyline operator} introduced by B{\"o}rzs{\"o}nyi et al. \cite{BorzsonyiKS01} to the database research community. Since then, this operator has received lots of attention and is studied extensively in multi-criteria decision making applications (\cite{PapadiasTFS03}, \cite{LiOTW06}, \cite{Sharifzadeh06}, \cite{DellisS07}, \cite{LinYZZ07}, \cite{WuTWDY09}, \cite{WuXMH09} for survey). Given a dataset of products $\mathcal{P}$, the standard skyline query returns all products $p\in\mathcal{P}$ that are not dominated by any other products $p^\prime\in\mathcal{P}$. A product $p$ is considered to dominate another product $p^\prime$ \emph{iff} it is as good as $p^\prime$ in every aspects of $p^\prime$, but better than $p^\prime$ in at least one aspect of $p^\prime$. Mathematically, $p$ dominates $p^\prime$, denoted by $p\prec p^\prime$, iff: (i) $\forall i\in \{1,2,..., d\}$, $p^i\le p^{\prime i}$ and (ii) $\exists j\in \{1,2,...,d\}$, $p^j<p^{\prime j}$, assuming that smaller values are preferred in all dimensions, $p^i$ and $p^{\prime i}$ denote the $i$th dimensional values of $p$ and $p^\prime$, respectively and $\mathcal{P}$ is a set of $d$-dimensional data objects. For example, consider the dataset of wine products given in Fig. \ref{fig:exampledataset}(a), the standard skyline operator \cite{BorzsonyiKS01} on this wine dataset returns $\{w_1, w_2\}$ as no other wine can dominate these wines in terms of \emph{1- percentage of grape juice content} (1-GraCon(\%)) and \emph{price}(\$).

Though standard skyline queries \cite{BorzsonyiKS01} can trade-off well if there are multiple dimensions of a product and a customer is unable to weight these dimensions, not all customers may prefer to minimize/maximize every dimensional value of a product, rather s/he may like certain range for it, e.g., laptop screen size, GraCon(\%) etc. To address this, Papadias et al. \cite{PapadiasTFS03} propose \emph{dynamic skyline query}, which retrieves data objects $p$ that are not \emph{dynamically dominated} by another data object $p^\prime$ w.r.t. a customer preference $c$, where $c$ is also a $d$-dimensional data object. Unlike standard skyline queries \cite{BorzsonyiKS01} where the aspects of $p$ is directly compared with the corresponding aspects of $p^\prime$ without considering any customer object, the dynamic skyline query compares the absolute differences of the aspects of $p$ and the customer object $c$ with the corresponding absolute differences of the aspects of $p^\prime$ and the customer object $c$ in deciding the dominance between $p$ and $p^\prime$. Mathematically, a data object $p$ dynamically dominates another data object $p^\prime$ w.r.t. a customer object $c$, denoted by $p\prec_c p^\prime$, iff: (i) $\forall i\in \{1,2,..., d\}$, $|p^i-c^i|\le |p^{^\prime i}-c^i|$ and (ii) $\exists j\in \{1,2,...,d\}$, $|p^j-c^j|<|p^{\prime j}-c^j|$. For example, consider the dataset of wines given in Fig. \ref{fig:exampledataset}(a) and the customer preferences in Fig. \ref{fig:exampledataset}(b), the dynamic skyline query of $c_1$ on the wine dataset returns $w_3$ as no other wines can dominate $w_3$ in view of $c_1$, i.e., $w_3$ matches the customer preference $c_1$ better than any other wines given in Fig. \ref{fig:exampledataset}(a).


Both the standard skyline \cite{BorzsonyiKS01} and dynamic skyline \cite{PapadiasTFS03} queries retrieve data objects from $P$ based on the customer's point of view, not the company's perspective. Dellis et al. \cite{DellisS07} present a new type of skyline queries, called \emph{reverse skyline}, which retrieves data objects from the company's point of view. Given a dataset of products $\mathcal{P}$, a set of customer preferences $\mathcal{C}$ and a product query $q$, the reverse skyline query retrieves all customers $c\in\mathcal{C}$ that include $q$ as one of their preferred products. Mathematically, given datasets $\mathcal{P}$ and $\mathcal{C}$ and a query $q$, a customer $c\in C$ is a reverse skyline of $q$, iff $\not \exists p\in\mathcal{P}$ such that (i) $\forall i\in \{1,2,...,d\}$, $|p^i-c^i|\le|q^i-c^i|$ and (ii) $\exists j\in \{1,2,...,d\}$, $|p^j-c^j|<|q^j-c^j|$. For example, consider the dataset of wine products given in Fig. \ref{fig:exampledataset}(a) and the customer preferences in Fig. \ref{fig:exampledataset}(b), the reverse skyline query of $w_1$ returns $c_2$ as no other wines in Fig. \ref{fig:exampledataset}(a) can dominate $w_1$ in view of $c_2$, i.e., $w_1$ is one of the preferred products of the the customer $c_2$. Like the standard and dynamic skyline queries, reverse skylines are also studied with great importance in the literature, specifically for measuring the \emph{influence} of a product and evaluating the market research queries (\cite{WuTWDY09}, \cite{ArvanitisDV12}, \cite{IslamLZ13}, \cite{IslamL16} for survey). 



Though the above skyline queries are important findings for studying the customer-product relationships over certain data, none of them is applicable over uncertain data. In works \cite{LianC08}, \cite{LianC10} Lian et al. present a threshold-based approach for evaluating reverse skyline queries over uncertain data. To find the threshold-based reverse skyline of a probabilistic product $p\in\mathcal{P}$, the authors first discover the probable alternative products of a customer $c\in\mathcal{C}$, called \emph{probabilistic dynamic skyline}. The probabilistic dynamic skyline of a customer $c$, denoted by $PDS(c)$, is computed as follows: $\{\forall p\in\mathcal{P}|Pr^c_{DSky}(p)\ge\delta\}$, where $Pr^c_{DSky}(p)$ denotes the \emph{dynamic skyline probability} of a product $p$ w.r.t. $c$ and is computed as follows: $Pr^c_{DSky}(p)=Pr(p)\times\prod_{\forall p^\prime\in\mathcal{P}\setminus\{p\}, p^\prime\prec_{c} p}{(1-Pr(p^\prime))}$, $Pr(p)$ denotes the probability of $p$ and $\delta$ is a given threshold. Then, the \emph{probabilistic reverse skyline} of a product $p\in\mathcal{P}$, denoted by $PRS(p)$, consists of all customers $c\in\mathcal{C}$ that include $p$ in its probabilistic dynamic skyline, i.e, $\{\forall c\in\mathcal{C}| p\in PDS(c)\}$. For example, consider the wine products and the customers given in Fig. \ref{fig:exampledataset}. Assume that the popularity ratings in Fig. \ref{fig:exampledataset}(a) are the probabilities of wines. The probabilistic reverse skyline of $w_2$ retrieves customers $c_1$ and $c_2$ for $\delta\ge 0.48$. Certainly, the study of probabilistic reverse skylines \cite{LianC08}, \cite{LianC10} is an advancement for measuring the \emph{influence} of a product over uncertain data. However, these skylines are not that friendly from usability point of view. (\textbf{Friendliness}) One has to mention the \emph{threshold} $\delta$, which is certainly a burden. (\textbf{Stability}) The result set can also vary based on the settings of $\delta$ and therefore, is not stable. (\textbf{Fairness}) Furthermore, it is not favorable towards products with small dynamic skyline probabilities. 

Recently, Zhou et al. \cite{ZhouLXZL16} propose a new skyline query called \emph{uncertain dynamic skyline} to compute the probable alternative choices for a customer. Unlike probabilistic dynamic skyline \cite{LianC08}, \cite{LianC10}, the uncertain dynamic skyline \cite{ZhouLXZL16} is stable and one does not need to provide any \emph{threshold} value. A product $p\in\mathcal{P}$ is considered a member of the uncertain dynamic skyline of a customer $c$ as long as $\not\exists p^\prime\in\mathcal{P}$ such that (i) $p^\prime\prec_c p$ and (ii) $Pr^c_{DSky}(p^\prime) \ge Pr^c_{DSky}(p)$. For example, consider the dataset of probabilistic wine products and the customer preferences given in Fig. \ref{fig:exampledataset}, the uncertain dynamic skyline of $c_1$, denoted by $UDS(c_1)$, retrieves $w_2$ and $w_3$, as no other wines can dynamically dominate them or their dynamic skyline probabilities are greater than these two wines in view of $c_1$. To compute the \emph{influence} of a probabilistic product $p\in\mathcal{P}$ through uncertain dynamic skyline, one has to compute the uncertain dynamic skyline of each customer $c\in\mathcal{C}$, i.e., $UDS(c)$ and then, check whether $UDS(c)$ includes $p$. As UDS query is computationally very expensive by itself, computing the influence of a probabilistic product via uncertain dynamic skyline \cite{ZhouLXZL16} is not efficient. 

This paper presents a new skyline query, called \emph{uncertain reverse skyline}, for measuring the \emph{influence} of a product in uncertain data settings. We also present efficient pruning ideas and an approach for processing the uncertain reverse skyline query of a probabilistic product. To be specific, our main contributions are as follows:


\begin{enumerate}[topsep=0pt,itemsep=-1ex,partopsep=1ex,parsep=1ex]
\item we introduce a novel skyline query, called \emph{uncertain reverse skyline}, for measuring the influence of a probabilistic product in uncertain data settings;
\item we present several pruning ideas and R-Tree data indexing based techniques to compute the uncertain reverse skyline and the influence score of a product in probabilistic databases;
\item we also present an efficient parallel computing approach for processing the uncertain reverse skyline query of a probabilistic product; and
\item finally, we demonstrate the efficiency of our approach by conducting extensive experiments with both real and synthetic datasets.
\end{enumerate}

The rest of the paper is organized as follows: Section \ref{sec:preliminaries} provides the preliminaries, Section \ref{sec:urs} presents the uncertain reverse skyline query and analyses the complexity of computing the influence score of probabilistic product through uncertain reverse skyline, Section \ref{sec:ourapproach} describes our approach in detail, Section \ref{sec:parallelapproach} presents our parallel approach, Section \ref{sec:experiments} presents the experimental results, Section \ref{sec:relatedwork} discusses the related work and finally, Section \ref{sec:conclusion} concludes the paper.

\section{Preliminaries}
\label{sec:preliminaries}

Consider a set of product objects $\mathcal{P}$ and a set of customer preferences $\mathcal{C}$, where a product object $p\in\mathcal{P}$ and a customer preference $c\in\mathcal{C}$ are $d$-dimensional points modeled as $<p^1, p^2,...,p^d>$ and $<c^1, c^2,..., c^d>$, respectively. The $p^i$ denotes the value of the product $p$ in the $i$th dimension, whereas the $c^i$ denotes the preferred value of the customer $c$ in the $i$th dimension of a product. If the product objects $p\in\mathcal{P}$ are associated with a probability (e.g., popularity rating), then we call it probabilistic product set. The probability of a product $p\in\mathcal{P}$ is denoted by $Pr(p)$. We use product and product object as well as customer and customer preference interchangeably. The query object, denoted by $q$, can represent both a product and a customer.


\begin{definition} \textbf{Dynamic Dominance \cite{PapadiasTFS03}} A product $p\in\mathcal{P}$ dynamically dominates another product $p^\prime\in\mathcal{P}$ w.r.t. a customer $c$, denoted by $p\prec_c p^\prime$, iff the followings hold: (i) $\forall i\in\{1,2,...,d\}\text{ , }|p^i-c^i|\le|{p^\prime}^i-c^i|$ and (ii) $\exists j\in\{1,2,...,d\}\text{ , }|p^j-c^j|<|{p^\prime}^j-c^j|$.
\label{def:dd}
\end{definition}

\begin{example}
Consider the datasets of wine products $\mathcal{W}$ and the customer $c_1$. According to the Definition \ref{def:dd}, the wine product $w_3$ dominates the wine product $w_6$ w.r.t. the customer $c_1$, i.e., $w_3\prec_{c_1} w_6$.
\label{ex:dynamicdominance}
\end{example}

\begin{definition} \textbf{Dynamic Skyline Probability \cite{LianC10, ParkMS15}.}
The dynamic skyline probability of a product $p\in\mathcal{P}$ w.r.t. a customer $c$, denoted by $Pr^c_{DSky}(p)$, is computed as follows: 
\begin{equation}
Pr^c_{DSky}(p)=Pr(p)\times\prod_{\forall p^\prime\in\mathcal{P}\setminus\{p\}, p^\prime\prec_{c} p}{(1-Pr(p^\prime))}
\label{eq:dskyprob}
\end{equation}
\end{definition}

\begin{example}
Consider the probabilistic wine products $\mathcal{W}$ and customers $\mathcal{C}$ in Fig. \ref{fig:exampledataset}. As no other objects in  $\mathcal{W}$ dominates $w_3$ w.r.t. $c_1$, the dynamic skyline probability of $w_3$ w.r.t. $c_1$ is $Pr^{c_1}_{DSky}(w_3)=Pr(w_3)=0.40$. Since $w_3\prec_{c_1} w_6$, the dynamic skyline probability of $w_6$ w.r.t. $c_1$ is $Pr^{c_1}_{DSky}(w_6)=Pr(w_6)\times(1-Pr(w_3))=0.60\times(1-0.40)=0.36$.
\label{ex:dskyprobability}
\end{example}

\begin{lemma}
$Pr^c_{DSky}(p^\prime)<Pr^c_{DSky}(p)$ iff: (i) $p\prec_c p^\prime$ and (ii) $Pr(p^\prime)< Pr(p)\vee (Pr(p^\prime)\times (1-Pr(p)))< Pr(p)$ \cite{ZhouLXZL16}.
\label{lem:dskyrelationship}
\end{lemma}

\begin{definition} \textbf{Uncertain Dynamic Dominance (UD-Dominance) \cite{ZhouLXZL16}.} A probabilistic product $p\in\mathcal{P}$ UD-dominates another probabilistic product $p^\prime\in\mathcal{P}$ w.r.t. a customer $c$, denoted by $p\prec^u_c p^\prime$, iff the followings hold: (i) $p\prec_c p^\prime$ and (ii) $Pr^c_{DSky}(p)\ge Pr^c_{DSky}(p^\prime)$.
\label{def:udd}
\end{definition}

\begin{example}
Consider the probabilistic wine products $\mathcal{W}$ and customers $\mathcal{C}$ given in Fig. \ref{fig:exampledataset}. As $w_3\prec_{c_1} w_6$ (see Ex. \ref{ex:dynamicdominance}) and also, $Pr^{c_1}_{DSky}(w_3)>Pr^{c_1}_{DSky}(w_6)$ (see Ex. \ref{ex:dskyprobability}), $w_3$ UD-dominates $w_6$ w.r.t. $c_1$, i.e., $w_3\prec^u_{c_1} w_6$. 
\end{example}


\begin{definition} \textbf{Uncertain Dynamic Skyline (UDS) \cite{ZhouLXZL16}.}
Given a set of probabilistic products  $\mathcal{P}$ and a customer $c$, the uncertain dynamic skyline of $c$, denoted by $UDS(c)$, consists of all products $p\in\mathcal{P}$ such that $p$ is not UD-dominated by any other $p^\prime\in\mathcal{P}\setminus p$, w.r.t. $c$.  Mathematically, $UDS(c)=\{c\in\mathcal{P}|\not\exists p^\prime\in\mathcal{P}\setminus p: p^\prime\prec^u_c p\}$.
\label{def:uds}
\end{definition}

\begin{example}
Consider the probabilistic wine products $\mathcal{W}$ and the customers $\mathcal{C}$ given in Fig. \ref{fig:exampledataset}. According to Definition \ref{def:uds}, the uncertain dynamic skyline of the customer $c_1$, i.e., $UDS(c_1)$, consists of wines $w_2$ and $w_3$ as no other wines in $\mathcal{W}$ UD-dominates them w.r.t. $c_1$. Similarly, the $UDS(c_2)$ and $UDS(c_3)$ are $\{w_1, w_2\}$ and $\{w_3, w_5, w_6\}$, respectively.
\end{example}

\begin{definition} \textbf{Favorite Probability \cite{ZhouLXZL16}.} Given a probabilistic product set $\mathcal{P}$, the \emph{favorite probability} of a product $p$ in view of a customer $c$, denoted by $Pr^c_{Fav}(p)$, is computed as given as follows: 

\begin{equation}
Pr^c_{Fav}(p)=
    \begin{cases}      
      \frac{Pr^c_{DSky}(p)}{\sum_{\forall p^\prime\in UDS(c)}{Pr^c_{DSky}(p^\prime)}} & \text{ if }\ p\in UDS(c)\\
	0 & \text{otherwise}
    \end{cases}
\label{eq:favprob}
\end{equation}
\normalsize


The \textbf{favorability rating} of a product $p$ w.r.t. a customer set $\mathcal{C}$, denoted by $Pr^\mathcal{C}_{Fav}(p)$, is computed as follows: 

\begin{align}
Pr^\mathcal{C}_{Fav}(p)&=\sum_{\forall c\in\mathcal{C}}{Pr^c_{Fav}(p)} \nonumber \\ 
	&=\sum_{\forall c\in\mathcal{C}}{\frac{Pr^c_{DSky}(p)}{\sum_{\forall p^\prime\in UDS(c)}{Pr^c_{DSky}(p^\prime)}}}
\label{eq:fabprobset}
\end{align}
\label{def:pis}
\end{definition}

\begin{example} Consider the datasets of probabilistic wine products $\mathcal{W}$ and the customers $\mathcal{C}$ as given in Fig. \ref{fig:exampledataset}. The favorability rating of $w_1$ w.r.t. the customer set $\mathcal{C}$ is $Pr^\mathcal{C}_{Fav}(w_1)$ $=\frac{0.00}{0.48+0.40}+\frac{0.90}{0.90+0.80}+\frac{0.00}{0.40+0.42+0.60}=0.53$. Similarly, the favorability rating of $w_2$ w.r.t. $\mathcal{C}$ is $Pr^\mathcal{C}_{Fav}(w_2)=\frac{0.48}{0.48+0.40}+\frac{0.90}{0.90+0.80}+\frac{0.00}{0.40+0.42+0.60}=1.02$.

\label{ex:favprob}
\end{example}

\section{Uncertain Reverse Skyline}
\label{sec:urs}

Here, we present a new skyline query, called \emph{uncertain reverse skyline query} based on \emph{UD-Dominance}\cite{ZhouLXZL16}. 

\begin{definition} \textbf{Uncertain Reverse Skyline (URS).} Given a set of probabilistic products  $\mathcal{P}$, a set of customers $\mathcal{C}$ and a query product $q$, the uncertain reverse skyline of $q$, denoted by $URS(q)$, consists of all customers $c\in\mathcal{C}$ such that $q$ appears in $UDS(c)$, i.e., $q\in UDS(c)$. Mathematically, a customer $c\in\mathcal{C}$ appears in $URS(q$) iff $\not\exists p\in\mathcal{P}$ such that: (a) $p\prec_c q$ and (b) $Pr^c_{DSky}(p)\ge Pr^c_{DSky}(q)$.
\label{def:urs}
\end{definition}

\begin{example}
Consider the datasets of probabilistic wines $\mathcal{W}$ and the customers $\mathcal{C}$ given in Fig. \ref{fig:exampledataset}. According to Definition \ref{def:urs}, the $URS(w_1)$ consists of $c_2$ only. The $URS(w_2)$ and $URS(w_3)$ are $\{c_1, c_2\}$ and $\{c_1, c_3\}$, respectively.
\end{example}

Unlike the probabilistic reverse skyline \cite{LianC10}, \cite{LianC13}, the uncertain reverse skyline proposed here is \emph{user friendly}, \emph{stable} and \emph{fair}. One does not need to provide the setting of threshold $\delta$ for computing the uncertain reverse skyline and it does not favor the query product over another one unless the query product strictly dominates the other one and the dynamic skyline probability of the query product is better than the other one. The uncertain reverse skyline always returns the same result, i.e., there is no threshold dependency.

\begin{definition} \textbf{Influence.} The \emph{influence set} of a probabilistic product $p\in\mathcal{P}$, denoted by $IS(p)$, consists of all customers $c\in\mathcal{C}$ that appear in the uncertain reverse skyline of $p$, i.e, $IS(p)=URS(p)$. Given a set of probabilistic products $\mathcal{P}$ and the customer set $\mathcal{C}$, the influence score of a probabilistic product $p$, denoted by $\tau(p)$, is measured by its favorability rating w.r.t. $\mathcal{C}$, i.e., $\tau(p)=Pr^\mathcal{C}_{Fav}(p)$.
\label{def:influence}
\end{definition}


\begin{example} Consider the datasets of probabilistic wine products $\mathcal{W}$ and the customers $\mathcal{C}$ as given in Fig. \ref{fig:exampledataset}. The influence score of wine product $w_1$ is $\tau(w_1)=Pr^\mathcal{C}_{Fav}(w_1)=0.53$ (easy to verify from Ex. \ref{ex:favprob}). Similarly, the influence score of wine product $w_2$ is $\tau(w_2)=Pr^\mathcal{C}_{Fav}(w_2)=1.02$. 
\end{example}

\subsection{Complexity Analysis}
\label{sec:complexity}

A naive approach of computing the \emph{influence score} of a product $p\in\mathcal{P}$ like the one proposed by Zhou et al. \cite{ZhouLXZL16} first computes the uncertain dynamic skyline of each customer $c\in\mathcal{C}$ and then, check whether the $UDS(c)$ includes the product $p$ and then computes its \emph{influence score} by following Eq. \ref{eq:fabprobset}. However, this approach requires the computation of $|\mathcal{C}|$ uncertain dynamic skylines, i.e., $UDS(c), \forall c\in\mathcal{C}$. As the \emph{UDS} query itself is computationally prohibitive, this na\"ive approach is not efficient enough to compute the influence score of a product $p$, i.e., $\tau(p)$. The following lemma guides how to efficiently compute $\tau(p)$ through the uncertain reverse skyline of $p$, i.e, $URS(p)$.  

\begin{lemma}
$\tau(p)=Pr^{URS(p)}_{Fav}(p)=\sum_{\forall c\in URS(p)}{Pr^c_{Fav}(p)}$.
\label{lem:influencescorenew}
\end{lemma}

\begin{proof} From Definition \ref{def:influence} and Eq. \ref{eq:fabprobset}, we get:

\begin{align*}
\tau(p)&=Pr^\mathcal{C}_{Fav}(p)
&=\sum_{\forall c\in\mathcal{C}}{\frac{Pr^c_{DSky}(p)}{\sum_{\forall p^\prime\in UDS(c)}{Pr^c_{DSky}(p^\prime)}}}
\end{align*}

Now, we can divide the customers $c\in\mathcal{C}$ in view of the product $p$ into two groups: (a) the customers $c\in\mathcal{C}$ that appear in the uncertain reverse skyline of $p$, i.e., $URS(p)$ and (b) the rest, i.e., $\mathcal{C}\setminus URS(p)$. Therefore, we can rewrite the above as given as follows:

\begin{align*}
\tau(p)&=\sum_{\forall c\in URS(p)}{\frac{Pr^c_{DSky}(p)}{\sum_{\forall p^\prime\in UDS(c)}{Pr^c_{DSky}(p^\prime)}}}\\
&+\sum_{\forall c^\prime\in\{\mathcal{C}\setminus URS(p)\}}{\frac{Pr^{c^\prime}_{DSky}(p)}{\sum_{\forall p^\prime\in UDS(c^\prime)}{Pr^{c^\prime}_{DSky}(p^\prime)}}}\\
 &=\sum_{\forall c\in URS(p)}{Pr^c_{Fav}(p)}+\sum_{\forall c^\prime\in \{\mathcal{C}\setminus URS(p)\}}{Pr^{c^\prime}_{Fav}(p)}
\end{align*}

According to Definition \ref{def:urs}, a product $p$ does not appear in the uncertain dynamic skyline of a customer $c^\prime$ if $c^\prime\not\in URS(p)$. Therefore, we get $Pr^{c^\prime}_{Fav}(p)=0$, $\forall c^\prime\in\mathcal{C}\setminus URS(p)$ and the above can be rewritten as given as follows:

\begin{align}
\tau(p)&=\sum_{\forall c\in URS(p)}{Pr^c_{Fav}(p)}+0\nonumber\\
\tau(p)         &=Pr^{URS(p)}_{Fav}(p)
\label{eq:ursscore}
\end{align}

Hence, the lemma, i.e., $\tau(p)=Pr^{URS(p)}_{Fav}(p)$. 
\end{proof}


From Lemma \ref{lem:influencescorenew}, we conclude that the efficiency of computing the \emph{influence score} of a product depends merely on the efficiency of computing its \emph{uncertain reverse skyline}, i.e., $URS(p)$. We present efficient pruning ideas and R-Tree data indexing based techniques for processing the uncertain reverse skyline query of a product in Section \ref{sec:ourapproach}. As we experience voluminous product and customer data in most data retrieval systems these days, we also present a parallel uncertain reverse skyline query evaluation technique in Section \ref{sec:parallelapproach}, which outperforms its serial counterparts significantly.

\section{Our Approach}
\label{sec:ourapproach}

This section presents our pruning ideas and the detail of uncertain reverse skyline query processing techniques based on probabilistic R-Tree data indexing. 

\subsection{Pruning Ideas}

\begin{definition}
\textbf{Orthant.} Given an object $p$ and a query $q$, the orthant $O$ of $p$ w.r.t. $q$, denoted by $O_q(p)$, is computed as: $O_q^i(p)=0$ \emph{iff} $p^i\le q^i$, otherwise $O_q^i(p)=1$. 
\end{definition}

A $d$-dimensional query $q$ has $2^d$ orthants in total, e.g., the orthants of $w_1$ and $w_2$ are shown as red-colored binary strings in Fig. \ref{fig:ursofw1andw2}(a) and Fig. \ref{fig:ursofw1andw2}(b), respectively.

\begin{definition} \textbf{Midpoint.} The midpoint $m$ of a product $p$ w.r.t. a query product $q$ is computed as given as follows: $m^i = (p^i + q^i)/2, \forall i\in\{1,2,...,d\}$. 
\end{definition}

\begin{example}
Consider the datasets of wine products $\mathcal{W}$ as given in Fig. \ref{fig:exampledataset}(a). The midpoint of $w_6$ w.r.t. the query product $w_1$ is $m_6=<55, 75>$. Similarly the midpoints of $w_2$, $w_3$ and $w_4$ w.r.t. $w_1$ are $m_2=<30, 80>$, $m_3=<50, 120>$ and $m_4=<35, 145>$, respectively. These midpoints are depicted in Fig. \ref{fig:ursofw1andw2}(a).
\end{example}

\begin{lemma}
Assume $m^\prime$ is a midpoint of $p^\prime$ w.r.t. $p$ and the followings hold: (i) $O_p(m^\prime)=O_p(c)$; (ii) $m^\prime\prec_p c$; and (iii) $Pr(p)< Pr(p^\prime)\vee (Pr(p)\times (1-Pr(p^\prime)))< Pr(p^\prime)$. Then, we get $Pr^c_{DSky}(p)<Pr^c_{DSky}(p^\prime)$ and $c\not\in URS(p)$.
\label{lem:midpointfilteringURS}  
\end{lemma}

\begin{proof}
As $m^\prime$ is a midpoint of $p^\prime$ w.r.t. the product $p$, we get $p^\prime\prec_{c} p \leftrightarrow m^\prime\prec_{p}c$ iff $O_p(m^\prime)=O_p(c)$ \cite{WuTWDY09}. This satisfies the conditions given for Lemma \ref{lem:dskyrelationship}, i.e., $Pr^c_{DSky}(p)<Pr^c_{DSky}(p^\prime)$ if conditions (i)-(iii) hold. Now, we get $c\not\in URS(p)$ according to Definition \ref{def:urs} as $p^\prime\prec_c p$ and $Pr^c_{DSky}(p)<Pr^c_{DSky}(p^\prime)$. Hence, the lemma.
\end{proof}

\begin{figure}[tb]
\centering
\setlength{\tabcolsep}{2pt}
\centering
\begin{tabular}{cc}
\begin{minipage}[a]{0.48\linewidth}
\subfigure[]{
\includegraphics[scale=0.18]{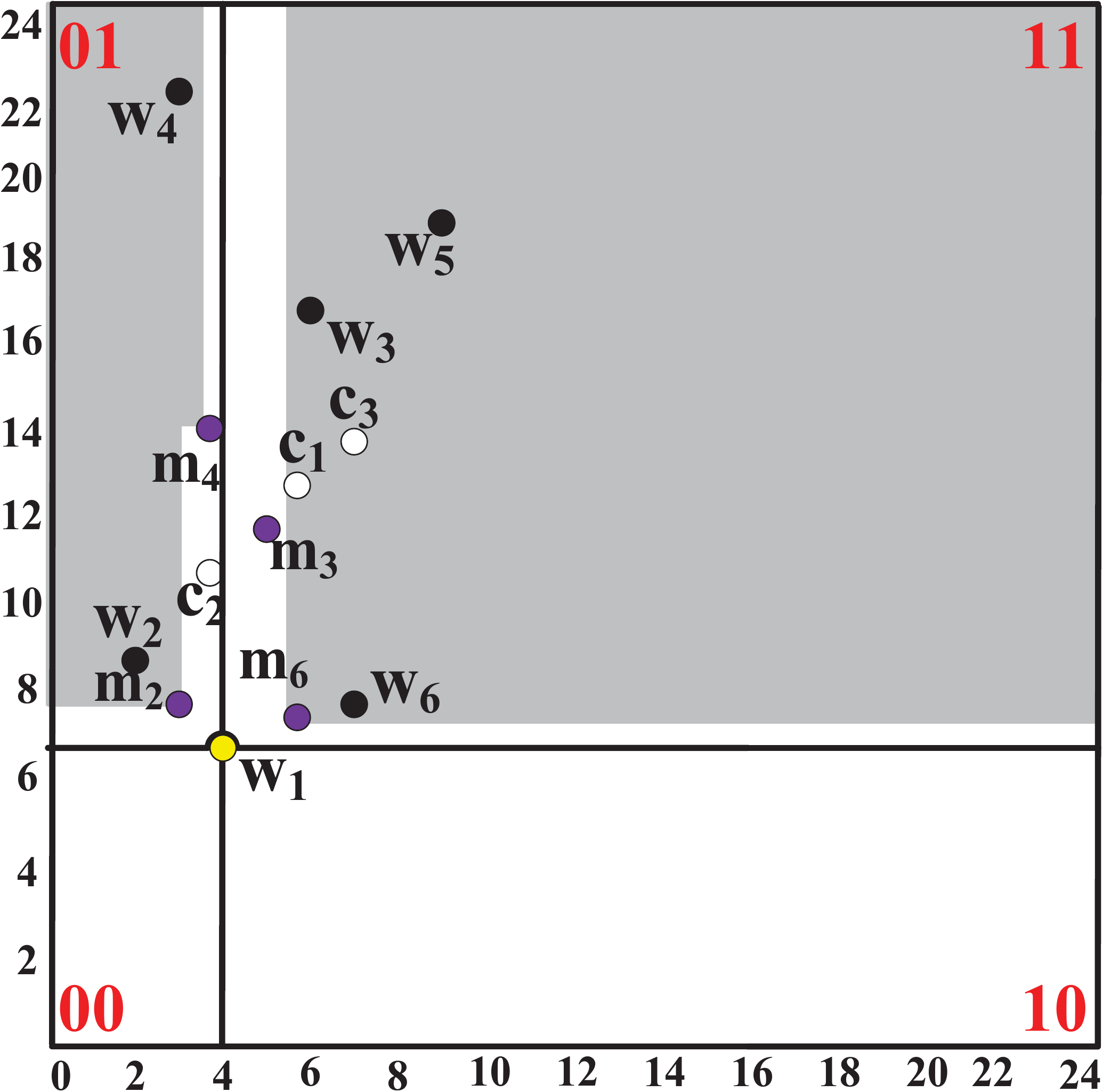}
}
\end{minipage}
&
\begin{minipage}[a]{0.48\linewidth}
\subfigure[]{
\includegraphics[scale=0.18]{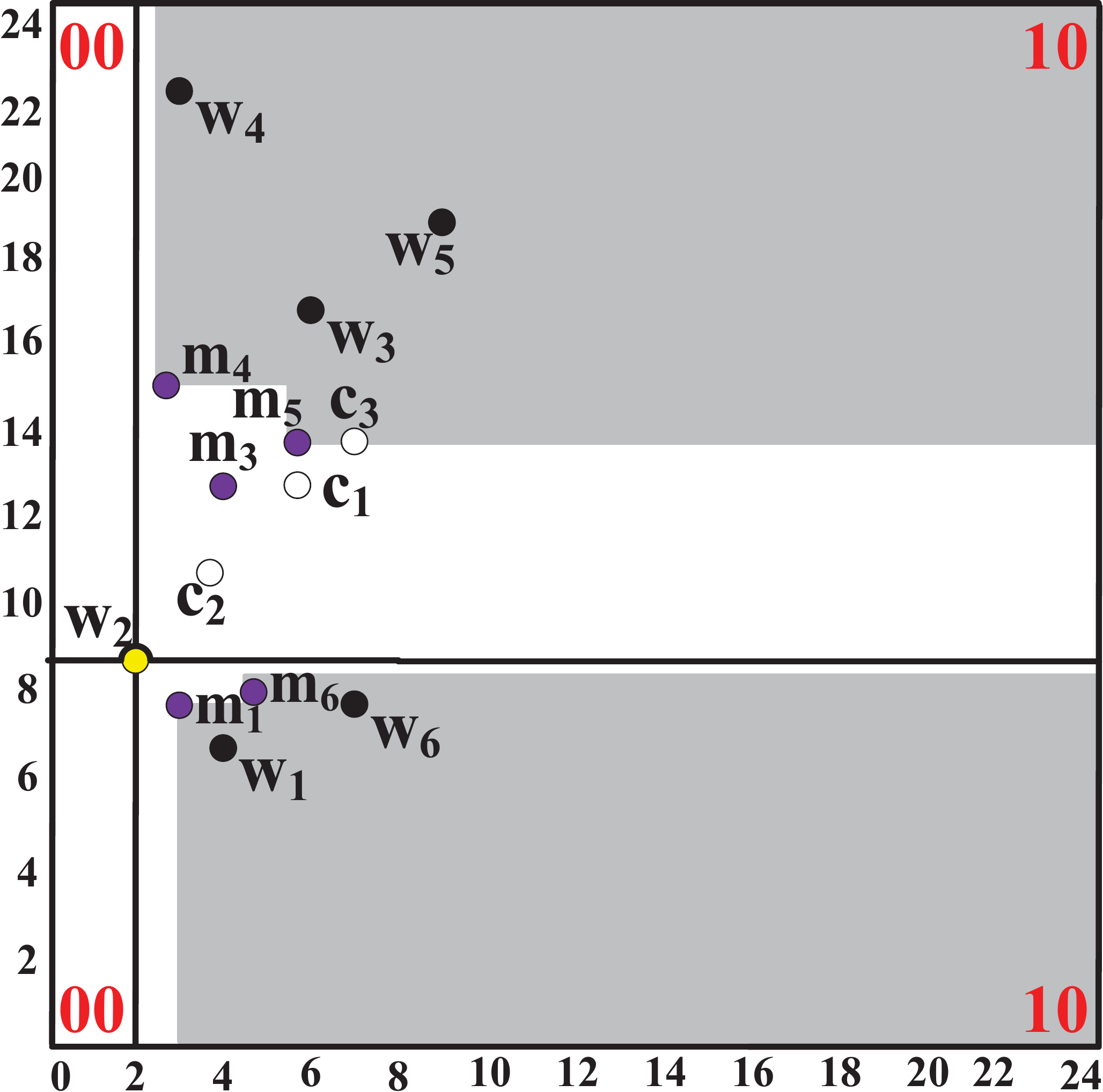}
}
\end{minipage}
\end{tabular}
\caption{UD-Dominance Regions (UDRs) of (a) $w_1$ and (b) $w_2$}
\label{fig:ursofw1andw2}
\end{figure}

\begin{definition} \textbf{UD-Dominance Region (UDR).} Given a set of probabilistic products $\mathcal{P}$ in a $d$-dimensional data space, a region is said to be a \emph{UD-dominance region} of a product $p\in\mathcal{P}$, denoted by $UDR(p)$, for which $\forall c\in UDR(p)$, $\exists p^\prime\in\mathcal{P}$ such that the followings hold: (i) $O_p(m^\prime)=O_p(c)$, (ii) $m^\prime\prec_p c$ and (iii) $Pr_{DSky}^c(p)\le Pr_{DSky}^c(p^\prime)$, where $m^\prime$ is the midpoint of the product $p^\prime$ w.r.t. $p$. 
\label{def:uddominanceregion}
\end{definition}

\begin{example}
Consider the datasets of probabilistic wine products $\mathcal{W}$ and the customers $\mathcal{C}$ as given in Fig. \ref{fig:exampledataset}. The UD-dominance regions of $w_1$ and $w_2$ are shown as gray patterned regions in Fig. \ref{fig:ursofw1andw2}(a) and Fig. \ref{fig:ursofw1andw2}(b), respectively. Here, the $UDR(w_1)$ is defined by the midpoints of $w_2$, $w_4$ and $w_6$ w.r.t. $w_1$. Similarly, the $UDR(w_2)$ is defined by the midpoints of $w_1$, $w_4$, $w_5$ and $w_6$ w.r.t. $w_2$.  
\end{example}

\begin{lemma}
A customer $c\in UDR(p)$ is not an uncertain reverse skyline of $p$, i.e., $c\not\in URS(p)$ if $c\in UDR(p)$.
\label{lem:udr}
\end{lemma}

\begin{proof}
Assume that $c\in UDR(p)$. According to the Definition \ref{def:uddominanceregion}, $\exists p^\prime\in\mathcal{P}$ such that the midpoint of $p^\prime$ w.r.t. $p$ dominates $c$ w.r.t. $p$, i.e., $p^\prime\prec_c p$ (conditions (i)-(ii)) and also,  $Pr_{DSky}^c(p)\le Pr_{DSky}^c(p^\prime)$. Therefore, the $UDS(c)$ does not include $p$ according to Definition \ref{def:uds}, which implies $c\not\in URS(p)$ according to Definition \ref{def:urs}. 
\end{proof}

\begin{definition} \textbf{Uncertain Midpoint Skyline.} Given a set of probabilistic products $\mathcal{P}$, the uncertain midpoint skyline of a probabilistic query product $q$, denoted by $UMSL(q)$, consists of a minimal set of midpoints of the products $p\in\mathcal{P}$ that defines the \emph{UD-dominance region} of $q$. 
\label{def:uncertainmidpointskyline}
\end{definition}

\begin{lemma}
If there are two products $p\in\mathcal{P}$ and $p^\prime\in\mathcal{P}$ such that the following holds: (i) $m\prec_q m^\prime$ and (ii) $\forall c\in\mathcal{C}$, $m^\prime\prec_q c\rightarrow m\prec_q c$ and $Pr_{DSky}^c(p)>Pr_{DSky}^c(q)$, then $m^\prime\not\in UMSL(q)$ but $m\in UMSL(q)$, where $m$ and $m^\prime$ are the midpoints of the products $p$ and $p^\prime$ w.r.t. $q$, respectively.
\label{lem:umslcorrectness}
\end{lemma}

\begin{proof}
Assume that $\exists c\in\mathcal{C}$ such that $m^\prime\prec_q c$ and $Pr_{DSky}^c(p^\prime)>Pr_{DSky}^c(q)$, where $m^\prime$ is the midpoint of the product $p^\prime\in\mathcal{P}$, but $m^\prime\not\in UMSL(q)$. This can not happen. Either $m^\prime\in UMSL(q)$ or $\exists m\in UMSL(q)$ such that $m\prec_q m^\prime$ and $Pr_{DSky}^c(p)>Pr_{DSky}^c(q)$, where $m$ is the midpoint of a product $p\in\mathcal{P}$ and $p\ne p^\prime$. For the former case, the $UMSL(q)$ is already correct as $c$ will be pruned by $m^\prime$ from $URS(q)$. For the later case, we get $m\prec_q c$ as $m\prec_q m^\prime$ and $m^\prime\prec_q c$ (transitivity of dominance). Since $Pr_{DSky}^c(p)>Pr_{DSky}^c(q)$, $c$ can be pruned from $URS(q)$ by $m$ even if $m^\prime\not\in UMSL(q)$. Hence, the lemma.
\end{proof}

\begin{example}
The uncertain midpoint skyline of $w_1$ consists of the midpoints of the products $w_2$, $w_4$ and $w_6$ w.r.t. $w_1$, i.e., $UMSL(w_1)=\{m_2, m_4, m_6\}$, where $m_2$, $m_4$, $m_6$ are the midpoints of $w_2$, $w_4$ and $w_6$ w.r.t. $w_1$. Similarly, the $UMSL(w_2)$ consists of the midpoints of the wine products $w_1$, $w_4$, $w_5$ and $w_6$ w.r.t. $w_2$.
\end{example}

\subsection{Data Indexing}
\label{sec:prtree}
From Lemma \ref{lem:udr} and Lemma \ref{lem:umslcorrectness}, it is obvious that we need to compute the $UMSL(q)$ of a probabilistic product $q$ to compute its uncertain reverse skyline. Thats is, the midpoints of the probabilistic products $p\in\mathcal{P}$ that defines the UD-domiance region of the query product $q$. This section presents an efficient approach to approximate the UD-dominance region of a probabilistic product by extending the R-Tree \cite{Guttman84} based data indexing for probabilistic product databases, called PR-Tree, which can take advantage of Lemma \ref{lem:udr} to compute its uncertain reverse skyline. The idea of PR-Tree is to augment each R-Tree node with the maximum and minimum probabilities of its children and store these probabilities in the tree node along with the links to its children. To construct the PR-Tree, we convert each product $p\in\mathcal{P}$ to its corresponding midpoint $m$ and then, insert it in the tree. We also index the customer data by the general R-Tree, which is refereed as CR-Tree in this paper. We use R-Tree to denote either of the trees throughout this paper. In connection with computing the uncertain reverse skyline of a product $q$ using R-Tree, we make the following statements.

\begin{itemize}
\item A midpoint $m$ is said to have the same orthant as an R-Tree node $n$, denoted by $O_q(m)=O_q(n)$, if all $2^d$ corners of node $n$ have the same orthant w.r.t. $q$ as $m$ does w.r.t. $q$.

\item An object $m$ dynamically dominates a node $n$ w.r.t. a query object $q$, denoted by $m\prec_q n$, if all $2^d$ corners of $n$ is dynamically dominated by $m$ w.r.t. $q$.

\item The tree nodes are always accessed in order of their distances to the query product $q$.
\end{itemize}

\subsection{Query Processing}
This section describes how to process the uncertain reverse skyline query and the influence (score) of a probabilistic product through its uncertain reverse skyline in detail.

\subsubsection{Uncertain Reverse Skyline}
\label{sec:serialURS}

While computing the uncertain reverse skyline of a product $q$, we prune a PR-Tree node  as per the following lemma.

\begin{lemma}
A PR-Tree node $n$ is \emph{pruned} if $\exists m^\prime\in\mathcal{M}^\prime$ such that (i) $O_q(m^\prime)=O_q(n)$, (ii) $m^\prime\prec_q n$ and (iii) $Pr(q)<Pr(p^\prime)\vee Pr(q)\times(1-Pr(p^\prime))<Pr(p^\prime)$, where $\mathcal{M}^\prime$ is the set of midpoints of the products $\mathcal{P}^\prime\subseteq\mathcal{P}$ accessed so far in the PR-Tree while computing $UMSL(q)$ for $URS(q)$ and $m^\prime$ is the midpoint of the product $p^\prime\in\mathcal{P}^\prime$.
\label{ursPRTreeNodePruning}
\end{lemma}

\begin{proof}
As all $2^d$ corners of node $n$ has the same orthant w.r.t. $q$ as $m$ does w.r.t. $q$ (condition (i)) and any $m\in n$ is bounded by the corners of $n$, $m$ must have the same orthant w.r.t. $q$ as $m^\prime$ does. Also, as $m^\prime$ dynamically dominates $n$ w.r.t. $q$ and $m\in n$ is bounded by the corners of $n$, $m^\prime$ also dynamically dominates $m$ w.r.t. $q$, i.e., $m^\prime\prec_q m$. Therefore, $\forall c\in\mathcal{C}$, if $m\prec_q c$ and $Pr_{DSky}^c(p)\ge Pr_{DSky}^c(q)$, we also get $m^\prime\prec_q c$ and $Pr_{DSky}^c(p^\prime)\ge Pr_{DSky}^c(q)$ (condition (iii)), which implies $n$ can be \emph{pruned}. Hence, the lemma.
\end{proof}

While computing the uncertain reverse skyline of a product $q$, we prune a CR-Tree node  as per the following lemma.

\begin{lemma}
A CR-Tree node $n$ is \emph{pruned} if $\exists m\in UMSL(q)$ such that (i) $O_q(m)=O_q(n)$ and (ii) $m \prec_q n$.
\label{ursCRTreeNodePruning}
\end{lemma}

\begin{proof}
As all $2^d$ corners of node $n$ has the same orthant w.r.t. $q$ as $m$ does w.r.t. $q$ (condition (i)) and any $c\in n$ is bounded by the corners of $n$, $c$ must have the same orthant w.r.t. $q$ as $m$ does. Also, as $m$ dynamically dominates $n$ w.r.t. $q$ (condition (ii)) and $c\in n$ is bounded by the corners of $n$, $m$ dynamically dominates $c$ w.r.t. $q$, i.e., $m\prec_q c$. Therefore, $\exists p\in\mathcal{P}$ such that $p\prec_c q$, where $p$ is the corresponding product of the midpoint $m$ and $Pr_{DSky}^c(p)\ge Pr_{DSky}^c(q)$ as $m\in UMSL(q)$, which implies $n$ can be \emph{pruned}. Hence, the lemma.
\end{proof}

The steps of computing the uncertain reverse skyline of a product $q$, i.e., $URS(q)$, with R-Trees are listed as follows:

\begin{enumerate}[topsep=0pt,itemsep=-1ex,partopsep=1ex,parsep=1ex]
\item Firstly, we convert the products $p\in\mathcal{P}$ into their midpoints $m$ w.r.t. $q$ and index them into a PR-Tree.

\item We initialize $UMSL(q)$ to an empty set. Then, we retrieve the children of the root node of the PR-Tree and insert them into a mean-heap $\mathcal{H}^\mathcal{P}_q$. We repeatedly retrieve the front entry $E$ from $\mathcal{H}^\mathcal{P}_q$ until $\mathcal{H}^\mathcal{P}_q$ becomes empty and do the following: ignore $E$ iff $\exists m\in UMSL(q)$ such that (i) $O_q(m)=O_q(E)$ and (ii) $m\prec_q E$ (Lemma \ref{ursPRTreeNodePruning}), otherwise, insert its children into $\mathcal{H}^\mathcal{P}_q$ if $E$ is a non-leaf node, else add the midpoint $m$ contained in $E$ into $UMSL(q)$ iff $Pr(q)<Pr(p)\vee (Pr(q)\times(1-Pr(p)))<Pr(p)$, where $p$ is the corresponding product of the midpoint $m$ in $\mathcal{P}$.
  
\item We index the customer data into a CR-Tree and initialize $URS(q)$ to an empty set. Then, we retrieve the children of the root node of the CR-Tree and insert them into a mean-heap $\mathcal{H}^\mathcal{C}_q$. We repeatedly retrieve the front entry $E$ from $\mathcal{H}^\mathcal{C}_q$ until $\mathcal{H}^\mathcal{C}_q$ becomes empty and do the following: ignore $E$ iff $\exists m\in UMSL(q)$ such that (i) $O_q(m)=O_q(E)$ and $m\prec_q E$ (Lemma \ref{ursCRTreeNodePruning}), otherwise, insert its children into $\mathcal{H}^\mathcal{C}_q$ if $E$ is a non-leaf node, else add the $c$ contained in $E$ into $URS(q)$.
\end{enumerate}

The above steps are pseudocoded in Algorithm \ref{algo:serialURS}.

\begin{algorithm}[tb]
\scriptsize
\DontPrintSemicolon 
\SetKwInOut{Input}{Input}\SetKwInOut{Output}{Output}
\Input{$q$: query, $\mathcal{P}$: products, $\mathcal{C}$: customers}
\Output{$URS(q)$: uncertain reverse skyline of $q$}
\Begin
{  
     $\mathcal{M}\leftarrow$converProductsToMidpoints($\mathcal{P}$);\tcp*[r]{midpoints}
     $root^\mathcal{P}\leftarrow$constructPRTree($\mathcal{M}$);\tcp*[r]{create PR-Tree}

     $ UMSL(q)\leftarrow\emptyset$;\tcp*[r]{initialization}

     $\mathcal{H}^\mathcal{P}_q\leftarrow$ insert($\mathcal{H}^\mathcal{P}_q$, $children(root^\mathcal{P})$);\tcp*[r]{mean heap}
     \While{$\mathcal{H}^\mathcal{P}_q\ne\emptyset$}
        { 
             $E\leftarrow$retrieveFront$(\mathcal{H}^\mathcal{P}_q)$;\\    
 		\If{$\exists m\in UMSL(q):O_q(E)=O_q(m)$ and $m\prec_q E$}
		{         
			continue;\tcp*[r]{prune PR-Tree node as per Lemma \ref{ursPRTreeNodePruning}}
		}
	     \ElseIf{$!E$.isLeaf()} 
		{
			$\mathcal{H}^\mathcal{P}_q\leftarrow$ insert($\mathcal{H}^\mathcal{P}_q$, $children(E)$);\\
		} 
             \ElseIf{$Pr(q)<Pr(E)\vee(Pr(q)\times(1-Pr(E)))<Pr(E)$}
            {
                     $UMSL(q)\leftarrow$ add($UMSL(q)$, $E$);\tcp*[r]{$UMSL$ member}
            }            
        }

       $URS(q)\leftarrow\emptyset$;\\       
	 $root^\mathcal{C}\leftarrow$constructCRTree($\mathcal{C}$);\tcp*[r]{create CR-Tree}
	$\mathcal{H}^\mathcal{C}_q\leftarrow$ insert($\mathcal{H}^\mathcal{C}_q$, $children(root^\mathcal{C})$);\tcp*[r]{mean heap}

        \While{$\mathcal{H}^\mathcal{C}_q\ne\emptyset$}
        { 
		$E\leftarrow$retrieveFront$(\mathcal{H}^\mathcal{C}_q)$;\\    
 		\If{$\exists m\in UMSL(q):O_q(m)=O_q(E)$ and $m\prec_q E$}
		{         
			continue;\tcp*[r]{prune CR-Tree node as per Lemma \ref{ursCRTreeNodePruning}}
		}
		\ElseIf{$!E$.isLeaf()} 
		{
			$\mathcal{H}^\mathcal{C}_q\leftarrow$ insert($\mathcal{H}^\mathcal{C}_q$, $children(E)$);\\
		} 
             \Else
            {
                     $URS(q)\leftarrow$ add($URS(q)$, $E$);\tcp*[r]{member of $URS(q)$}
            }      
        }
}
\normalsize
\caption{Uncertain Reverse Skyline}
\label{algo:serialURS}
\end{algorithm}

\begin{lemma} 
Algorithm \ref{algo:serialURS} computes accurately the uncertain reverse skyline of an arbitrary probabilistic product $q$.
\label{lem:correctnessofURSAlgorithm}
\end{lemma}

\begin{proof}
The computation of the uncertain reverse skyline of $q$, i.e., $URS(q)$, starts scanning the products $p\in\mathcal{P}$, then converting them into their corresponding midpoints w.r.t. $q$ and thereafter, inserting them into the PR-Tree as given in lines 2-3. Then, we initialize $UMSL(q)$ to $\emptyset$ and insert the children of PR-Tree root into the min-heap $\mathcal{H}^\mathcal{P}_q$ in lines 4-5. The lines 6-13 repeatedly retrieve the front entry $E$ of $\mathcal{H}^\mathcal{P}_q$ until $\mathcal{H}^\mathcal{P}_q$ is empty and prune $E$ (PR-Tree node) as per Lemma \ref{ursPRTreeNodePruning}, otherwise, insert the children of $E$ into $\mathcal{H}^\mathcal{P}_q$ if $E$ is an internal node, else add the midpoint $m$ contained in $E$ (leaf node) into the $UMSL(q)$ only if $Pr(q)<Pr(p)\vee Pr(q)\times(1-Pr(p))<Pr(p)$ to make sure that if $\exists c\in\mathcal{C}$ such that $p\prec_c q$ and $Pr^c_{DSky}(p)>Pr^c_{DSky}(q)$ hold, $c$ can be pruned by $m$ as per Lemma \ref{lem:udr}, where $p$ is the corresponding product of $m$ in $\mathcal{P}$. As the entries (PR-Tree nodes) in $\mathcal{H}^\mathcal{P}_q$  are accessed in order of their distances to $q$, the $UMSL(q)$ computed in lines 6-13 is minimal and correct. Now, we initialize $URS(q)$ to $\emptyset$, constrcut CR-Tree of the customers $\mathcal{C}$ and insert the children of the CR-Tree root into the min-heap $\mathcal{H}^\mathcal{C}_q$ in lines 14-16. The lines 17-24 repeatedly retrieve the front entry $E$ of $\mathcal{H}^\mathcal{C}_q$ until $\mathcal{H}^\mathcal{C}_q$ is empty and prune $E$ (CR-Tree node) as per Lemma \ref{ursCRTreeNodePruning}, otherwise, insert the children of $E$ into $\mathcal{H}^\mathcal{C}_q$ if $E$ is an internal node, else add the customer $c$ contained in $E$ (leaf node) into the $URS(q)$ as per the Definition \ref{def:urs}. Hence, the lemma.
\end{proof}


\begin{algorithm}[tb]
\scriptsize
\DontPrintSemicolon 
\SetKwInOut{Input}{Input}\SetKwInOut{Output}{Output}
\Input{$q$: query, $\mathcal{P}$: products, $\mathcal{C}$: customers}
\Output{$\tau(q)$: influence score of $q$}
\Begin
{       
     $URS(q)\leftarrow$computeURS($q$);\tcp*[r]{Algorithm \ref{algo:serialURS}}
     $\tau(q)\leftarrow 0$;\tcp*[l]{initialization}
     \For{each $c\in URS(q)$}
        { 
             $UDS(c)\leftarrow$computeUDS($c$, $P\cup q$);\tcp*[r]{Approach in \cite{ZhouLXZL16}}
		$score\leftarrow 0$; \tcp*[l]{initialization}
		     \For{each $p\in UDS(c)$}
        		{ 
				$score\leftarrow score+Pr^c_{DSky}(p)$;\\
			}
 		 $\tau(q)\leftarrow \tau(q)+\frac{Pr^c_{DSky}(q)}{score}$; \tcp*[r]{Eq. \ref{eq:ursscore}}
        }
  
}
\normalsize
\caption{Influence Score}
\label{algo:serialInfluenceScore}
\end{algorithm}

\subsubsection{Influence Score}
\label{sec:influencescore}

As per Eq. \ref{eq:ursscore}, we need to compute the dynamic skyline probability of each product $p\in UDS(c)$ for each $c\in URS(q)$ to compute the influence score $\tau(q)$ of the query product $q$. To achieve this, we first compute the uncertain reverse skyline of $q$, i.e., $URS(q)$ by Algorithm \ref{algo:serialURS}. Then, we compute the dynamic skyline probability of each product $p\in UDS(c)$ for each $c\in URS(q)$ as per the approach proposed in \cite{ZhouLXZL16}. This idea is pseudocoded in Algorithm \ref{algo:serialInfluenceScore}.  Though, we adopt the approach proposed in \cite{ZhouLXZL16} for computing the dynamic skyline probability in Algorithm \ref{algo:serialInfluenceScore}, there is a significant difference between our approach and the approach proposed in \cite{ZhouLXZL16} for computing $\tau(q)$. The approach proposed in \cite{ZhouLXZL16} computes the $UDS(c)$ of each customer $c\in\mathcal{C}$ irrespective of whether $c$ is in $URS(q)$ or not to compute $\tau(q)$, which we don't do in our approach. Therefore, our approach is more efficient than the na\"ive approach proposed in \cite{ZhouLXZL16} for computing the influence score $\tau(q)$ of an arbitrary query product $q$.

\subsubsection{Optimization}
\label{sec:optimization}

Assume that $n_{far}$ is the farthest and $n_{near}$ is the nearest corner of a R-Tree node $n$ w.r.t. $q$ as shown by the green-colored bulleted objects in Fig. \ref{fig:RTreeNodeToNodePruning}. If $n$ is a PR-Tree node, also assume that $Pr(n_{far})=min\{Pr(p), \forall m\in n\}$ and $Pr(n_{near})=max\{Pr(p), \forall m\in n\}$, where $p$ is the corresponding product in $\mathcal{P}$ of the midpoint $m$. 

The following lemma guides how to prune a PR-Tree node by comparing it with another PR-Tree node while computing the uncertain midpoint skyline of an arbitrary query product $q$, i.e., $UMSL(q)$.

\begin{lemma}
A PR-Tree node $n^\prime$ can be pruned if $\exists n\in PR-Tree$ such that (i) $O_q(n)=O_q(n^\prime)$, (ii) $n_{far}\prec_q n^\prime_{near}$ and (iii) $Pr(q)<Pr(n_{far})\vee Pr(q)\times(1-Pr(n_{far}))<Pr(n_{far})$.
\label{lem:UMSLnodetonodepruning}
\end{lemma}

\begin{proof}
Assume that $\exists m^\prime\in n^\prime$ and $\exists c\in\mathcal{C}$ such that $m^\prime\prec_q c$ and $Pr^c_{DSky}(p^\prime)>Pr^c_{DSky}(q)$, where $p^\prime$ is the corresponding product in $\mathcal{P}$ of the midpoint $m^\prime$, i.e., $c\not\in URS(q)$. Now, there must exist a midpoint $m\in n$ such that $m\prec_q c$ because of conditions (i) and (ii) as follows: $m\prec_q n_{far}\wedge n_{far}\prec_q n^\prime_{near}\wedge n^\prime_{near}\prec_q m^\prime\wedge m^\prime\prec_q c\rightarrow m\prec_q c$ (transitivity of dominance). Now, $Pr(q)<Pr(p)\vee Pr(q)\times(1-Pr(p))<Pr(p)$ because of condition (iii), where $p$ is the corresponding product in $\mathcal{P}$ of the midpoint $m$, which implies $Pr^c_{DSky}(p)>Pr^c_{DSky}(q)$. Therefore, we can still prune $c$ by $m\in n$ even if we prune $n^\prime$. Hence, the lemma. 
\end{proof}

\begin{lemma}
The customers $c$ in a CR-Tree node $n$ can be safely added to $URS(q)$ if $\not\exists m\in UMSL(q)$ such that the followings hold: (i) $O_q(m)=O_q(n)$ and (ii) $m\prec_q n_{far}$.
\label{lem:URSnodecustomersaddition}
\end{lemma}

\begin{proof}

Assume that $\exists c\in n$ and the conditions (i)-(ii) are true, but $c\not\in URS(q)$. We prove that $URS(q)$ is incorrect. As $\not\exists m\in UMSL(q)$ such that $m\prec_q n_{far}$ and $c$ is bounded within the region of node $n$, we get $m\not\prec_q c$. Therefore, $c$ must be in $URS(q)$. Hence, the lemma.
\end{proof}

\begin{figure}[tb]
\centering
\setlength{\tabcolsep}{2pt}
\centering
\begin{tabular}{cc}
\begin{minipage}[a]{0.48\linewidth}
\subfigure[]{
\includegraphics[scale=0.14]{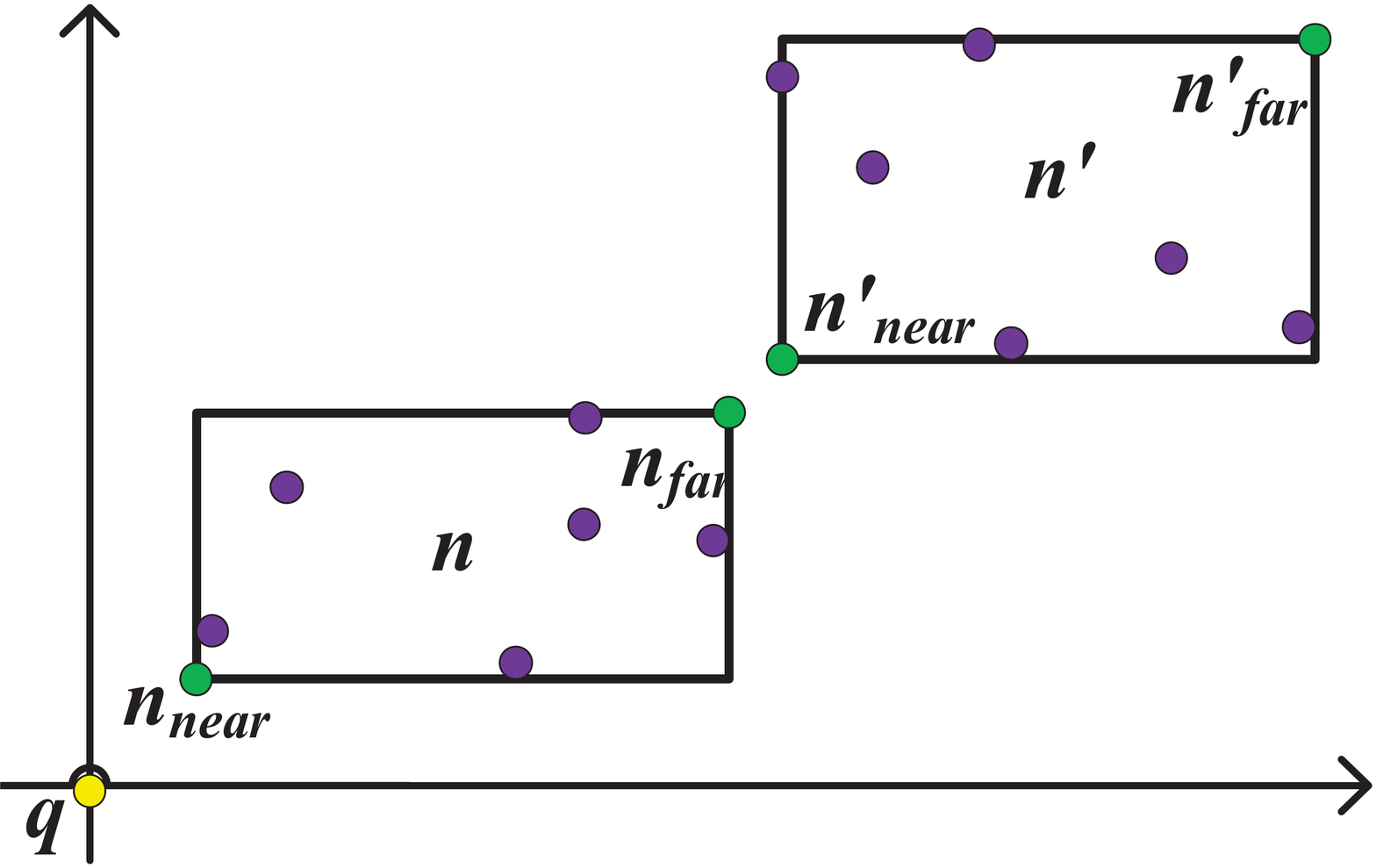}
}
\end{minipage}
&
\begin{minipage}[a]{0.48\linewidth}
\subfigure[]{
\includegraphics[scale=0.145]{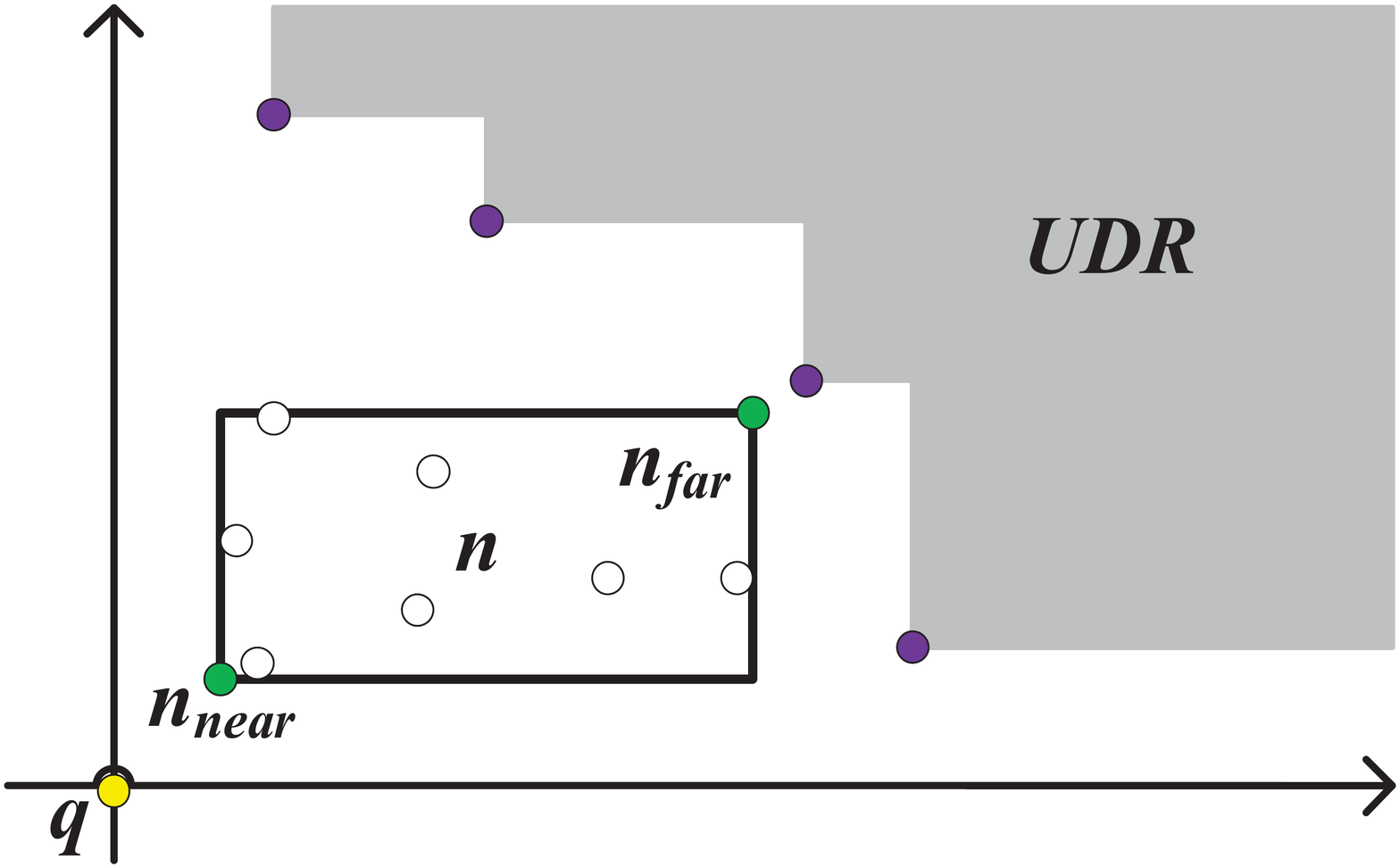}
}
\end{minipage}
\end{tabular}
\vspace{-2ex}
\caption{Optimization (a) a node prunes another node in PR-Tree node for $UMSL(q)$ and (b) adding each customers in a CR-Tree node into $URS(q)$, where purple-colored bulleted points represents midpoints of $q$}
\label{fig:RTreeNodeToNodePruning}
\vspace{-2ex}
\end{figure}

\begin{algorithm}[tb]
\scriptsize
\DontPrintSemicolon 
\SetKwInOut{Input}{Input}\SetKwInOut{Output}{Output}
\Input{$q$: query, $\mathcal{P}$: products, $\mathcal{C}$: customers}
\Output{$URS(q)$: uncertain reverse skyline of $q$}
\Begin
{  
     $\mathcal{M}\leftarrow$converProductsToMidpoints($\mathcal{P}$);\tcp*[r]{midpoints}
     $root^\mathcal{P}\leftarrow$constructPRTree($\mathcal{M}$);\tcp*[r]{create PR-Tree}

     $ UMSL(q)\leftarrow\emptyset$;\tcp*[r]{initialization}

     $\mathcal{H}^\mathcal{P}_q\leftarrow$ insert($\mathcal{H}^\mathcal{P}_q$, $children(root^\mathcal{P})$);\tcp*[r]{mean heap}
     \While{$\mathcal{H}^\mathcal{P}_q\ne\emptyset$}
        { 
             $E\leftarrow$retrieveFront$(\mathcal{H}^\mathcal{P}_q)$;\\    
 		\If{$\exists m\in UMSL(q):O_q(E)=O_q(m)$ and $m\prec_q E$}
		{         
			continue;\tcp*[r]{prune PR-Tree node as per Lemma \ref{ursPRTreeNodePruning}}
		}
	     \ElseIf{$!E$.isLeaf()} 
		{
			$\mathcal{H}^\mathcal{P}_q\leftarrow$ insert($\mathcal{H}^\mathcal{P}_q$, $children(E)$);\\
			$\mathcal{H}^\mathcal{P}_q\leftarrow$applyNodeToNodePruning($\mathcal{H}^\mathcal{P}_q$);\tcp*[r]{Lemma \ref{lem:UMSLnodetonodepruning}}
		} 
             \ElseIf{$Pr(q)<Pr(E)\vee(Pr(q)\times(1-Pr(E)))<Pr(E)$}
            {
                     $UMSL(q)\leftarrow$ add($UMSL(q)$, $E$);\tcp*[r]{$UMSL$ member}
            }            
        }

       $URS(q)\leftarrow\emptyset$;\\       
	 $root^\mathcal{C}\leftarrow$constructCRTree($\mathcal{C}$);\tcp*[r]{create CR-Tree}
	$\mathcal{H}^\mathcal{C}_q\leftarrow$ insert($\mathcal{H}^\mathcal{C}_q$, $children(root^\mathcal{C})$);\tcp*[r]{mean heap}

        \While{$\mathcal{H}^\mathcal{C}_q\ne\emptyset$}
        { 
		$E\leftarrow$retrieveFront$(\mathcal{H}^\mathcal{C}_q)$;\\    
 		\If{$\exists m\in UMSL(q):O_q(m)=O_q(E)$ and $m\prec_q E$}
		{         
			continue;\tcp*[r]{prune CR-Tree node as per Lemma \ref{ursCRTreeNodePruning}}
		}
		\ElseIf{$\not\exists m\in UMSL(q):O_q(m)=O_q(E)$ and $m\prec_q E_{far}$} 
		{
			 $URS(q)\leftarrow$ add($URS(q)$, $customers(E)$); \tcp*[r]{Lemma \ref{lem:URSnodecustomersaddition}}
		}
		\Else
		{
			$\mathcal{H}^\mathcal{C}_q\leftarrow$ insert($\mathcal{H}^\mathcal{C}_q$, $children(E)$);\tcp*[r]{non-leaf node}
		} 
        }
}
\normalsize
\caption{Optimized Uncertain Reverse Skyline}
\label{algo:optimizedSerialURS}
\end{algorithm}

The above optimization heuristics, i.e., Lemma \ref{lem:UMSLnodetonodepruning} and Lemma \ref{lem:URSnodecustomersaddition} are pseudocoded in Algorithm \ref{algo:optimizedSerialURS}. The difference between Algorithm \ref{algo:serialURS} and Algorithm \ref{algo:optimizedSerialURS} is that  Algorithm \ref{algo:optimizedSerialURS} applies PR-Tree node to node pruning on $\mathcal{H}^\mathcal{P}_q$ after inserting the children of an entry $E$ into $\mathcal{H}^\mathcal{P}_q$ while computing $UMSL(q)$ (lines 10-12) and adds the customers $c$ of a CR-Tree non-leaf node $E$ into $URS(q)$ if the conditions in Lemma \ref{lem:URSnodecustomersaddition} are satisfied without inserting the children into $\mathcal{H}^\mathcal{C}_q$ (lines 22-23). The optimization of influence score computation in Algorithm \ref{algo:serialInfluenceScore} is done by replacing Algorithm \ref{algo:serialURS} with Algorithm \ref{algo:optimizedSerialURS} in line 2 for computing the uncertain reverse skyline of $q$.

\section{Parallel Approach}
\label{sec:parallelapproach}

This section presents an efficient approach of computing the uncertain reverse skyline and the influence score of a product by parallelizing their evaluations for today's data intensive systems involving millions of customer objects.

\subsection{Computing Environment}
\label{sec:parallelcomputingenvironment}

We assume a simplified computing environment for evaluating uncertain reverse skyline queries in parallel in which a master processor, denoted by $\mathcal{T}_0$, is responsible for coordinating and managing the independent tasks carried out by the worker processors, denoted by $\{\mathcal{T}_j\}$. A worker processor $\mathcal{T}_j$ receives input data from the master and the task type, finishes the task accordingly and sends the processed result back to the master processor. The master processor may pre-process the input data before sending them to the workers. The master processor $\mathcal{T}_0$ finalizes the result in one or more rounds. We also assume that the communications and synchronizations between the master processor and the worker processors are integral part of this environment, and the computing powers of all worker processors are the same. 

\subsection{Parallel Uncertain Reverse Skyline}
\label{sec:parallelURS}

The parallel steps of computing the uncertain reverse skyline of a probabilistic product $q$, i.e., $URS(q)$, in two rounds are listed as follows:

\begin{enumerate}[topsep=0pt,itemsep=-1ex,partopsep=1ex,parsep=1ex]
\item In the first round, the master divides $\mathcal{P}$ into chunks $\mathcal{P}_j\subset\mathcal{P}$ (such that $\cup\mathcal{P}_j = \mathcal{P}$) and sends these chunks $\mathcal{P}_j$ and the query product $q$ to its workers.

\item A worker processor converts the products $p\in\mathcal{P}_j$ into their midpoints $m$ w.r.t. $q$ and index them into its local PR-Tree. Then, the worker computes the local uncertain midpoint skyline $UMSL_j$ by following the same technique as given in Step 2 in Section \ref{sec:serialURS}.
 
\item Then, the master does the followings: (i) collects all local $UMSL_j$s from its workers and insert them into a min heap $\mathcal{H}^\mathcal{P}_q$; (ii) initializes $UMSL(q)$ to $\emptyset$  and (iii) repeatedly retrieves the front entry $m$ from $\mathcal{H}^\mathcal{P}_q$ until it becomes empty and does the following: adds $m$ to $UMSL(q)$ if $\not\exists m^\prime\in UMSL(q)$ such that: $O_q(m^\prime)=O_q(m)$ and $m^\prime\prec_q m$, otherwise ignore $m$.

\item In the second round, the master divides $\mathcal{C}$ into chunks $\mathcal{C}_j\subset\mathcal{C}$ (such that $\cup\mathcal{C}_j = \mathcal{C}$) and sends these chunks $\mathcal{C}_j$ and the global $UMSL(q)$ to its workers.

\item A worker processor index $\mathcal{C}_j$ into its local CR-Tree. Then, the worker computes the local uncertain reverse skyline $URS_j$ by following the same technique as given in Step 3 in Section \ref{sec:serialURS}

\item Finally, the master collect all local $URS_j$s from its workers into the global $URS(q)$.
\end{enumerate}

\begin{algorithm}[t]
\scriptsize
\DontPrintSemicolon
\SetKwInOut{Input}{Input}\SetKwInOut{Output}{Output}
\Input{$q$: query, $\mathcal{P}$: products, $\mathcal{C}$: customers}
\Output{$URS(q)$: uncertain reverse skyline of $q$}
\Begin{

$\mathcal{T}_0:$partitionProductData($\mathcal{P}$); \tcp*[r]{\tiny{partition product data}}

$\mathcal{T}_0:$ \textbf{parallel}\text{ }\For{each $\mathcal{P}_j$}
	{
		sendQuery($\mathcal{T}_j$, $q$);\tcp*[r]{\tiny{send query product to $\mathcal{T}_j$}}             
		sendProducts($\mathcal{T}_j$, $\mathcal{P}_j$);\tcp*[r]{\tiny{send product subset to $\mathcal{T}_j$}}
		$\mathcal{T}_j:\mathcal{M}_j\leftarrow$converProductsToMidpoints($\mathcal{P}_j$);\tcp*[r]{\tiny{midpoints}}
	      $\mathcal{T}_j:root^\mathcal{P}_j\leftarrow$constructPRTree($\mathcal{M}_j$);\tcp*[r]{\tiny{create local PR-Tree}}		
		$UMSL_j\leftarrow \mathcal{T}_j:$ localMidpointSkyline($q$, $root^\mathcal{P}_j$);\\
	}

$\mathcal{T}_0:UMSL(q)\leftarrow$globalMidpointSkyline($q$, $\cup UMSL_j$);\\

$\mathcal{T}_0:$partitionCustomerData($\mathcal{C}$); \tcp*[r]{\tiny{partition customer data}}
$\mathcal{T}_0:$ \textbf{parallel}\text{ }\For{each $\mathcal{C}_j$}
	{
             sendGlobalMidpointSkyline($\mathcal{T}_j$, $UMSL(q)$);\tcp*[r]{\tiny{$UMSL(q)$}}
		sendCustomers($\mathcal{T}_j$, $\mathcal{C}_j$);\tcp*[r]{\tiny{send customer subset}}
		$\mathcal{T}_j:root^\mathcal{C}_j\leftarrow$constructCRTree($\mathcal{C}_j$);\tcp*[r]{\tiny{create local CR-Tree}}				
		$URS_j\leftarrow \mathcal{T}_j:$ localURS($q$, $root^\mathcal{C}_j$, $UMSL(q)$);
	}

$\mathcal{T}_0: URS(q)\leftarrow\cup URS_j$;\tcp*[r]{\tiny{global uncertain reverse skyline, $URS(q)$}}
}
\normalsize
\caption{Parallel Uncertain Reverse Skyline}
\label{algo:parallelURS}
\end{algorithm}

The above steps are pseudocoded in Algorithm \ref{algo:parallelURS} as explained below. The master processor $\mathcal{T}_0$ partitions the product data $\mathcal{P}$ equally for the workers in line 2. The master processor then sends the query product $q$ and the partitioned data $\mathcal{P}_j$ to the corresponding worker processor $\mathcal{T}_j$ in lines 4-5. In lines 6-8, the worker processor $\mathcal{T}_j$ converts $\mathcal{P}_j$ into the corresponding midpoints $\mathcal{M}_j$, constructs the local PR-Tree $root^\mathcal{P}_j$ and computes the local uncertain midpoint skyline $UMSL_j$ by calling \emph{localMidpointSkyline($q$, $root^\mathcal{P}_j$)} method which implements Step 2. Once computed, $\mathcal{T}_j$ sends the local $UMSL_j$ to the master $\mathcal{T}_0$ in line 8. The master $\mathcal{T}_0$ computes the global uncertain midpoint skyline $UMSL(q)$ by calling \emph{globalMidpointSkyline($q$, $\cup UMSL_j$)} method which implements Step 3) in line 9. The master processor $\mathcal{T}_0$ now partitions the customer data $\mathcal{C}$ equally for the workers in line 10 and then, sends the global $UMSL(q)$ and $\mathcal{C}_j$ to the corresponding worker $\mathcal{T}_j$ in lines 12-13. The worker processor $\mathcal{T}_j$ constructs the local CR-Tree $root^\mathcal{C}_j$ and computes the local $URS_j$ by calling method \emph{localURS($q$, $root^\mathcal{C}_j$, $UMSL(q)$)} which implements step 5 in lines 14-15. Finally, the local $URS_j$ are accumulated by the master $\mathcal{T}_0$ into the global uncertain reverse skyline $URS(q)$ in line 16 of Algorithm \ref{algo:parallelURS}.     

\begin{lemma}
The Algorithm \ref{algo:parallelURS} accurately computes the uncertain reverse skyline of an arbitrary query product $q$.
\end{lemma}

\begin{proof} 
Firstly, we prove that the global uncertain midpoint skyline, i.e., $UMSL(q)$ computed by Algorithm \ref{algo:parallelURS} is correct. The local midpoint skyline $UMSL_j$ of $q$ is correct for the partition $\mathcal{P}_i$ as we prove for $\mathcal{P}$ in Algorithm \ref{algo:serialURS}. Now, Algorithm \ref{algo:parallelURS} computes the global $UMSL(q)$ by accumulating the local $UMSL_j$s into the mean heap $\mathcal{H}^\mathcal{P}_q$ and thereafter, accessing the midpoints in $\mathcal{H}^\mathcal{P}_q$ in order of their distances to $q$. A midpoint $m$ is added to the global $UMSL(q)$ iff it's filtering capability cannot be achieved by another midpoint $m$ already existing in $UMSL(q)$. Therefore, the global $UMSL(q)$ can filter the customers $c\in\mathcal{C}$ that would be filtered by local $UMSL_j$s, i.e., the global $UMSL(q)$ is correct and minimal. Finally, the worker processor computes the local $URS_j$ for the customer set $c\in\mathcal{C}_j$ based on the global $UMSL(q)$ as we compute $URS(q)$ for $\mathcal{C}$ in Algorithm \ref{algo:serialURS}. As the selection of customers in the uncertain reverse skyline set of $q$ are mutually independent, the global $URS(q)$ accumulated in the master is correct. Hence, the lemma.
\end{proof}

\subsection{Parallel Influence Score}
\label{sec:parallelInfluenceScore}


This section presents an approach for computing the influence score of an arbitrary query product $q$ in parallel. More specifically, we parallelize the computation of the dynamic skyline probabilities of each product $p\in UDS(c)$ for each $c\in URS(q)$. Our approach is significantly different from the approach proposed in \cite{ZhouLXZL16}. The approach in \cite{ZhouLXZL16} computes the favorite probability $Pr^c_{Fav}(q)$ by executing the uncertain dynamic skyline query of each $c\in\mathcal{C}$ in different processing nodes without partitioning $\mathcal{P}$. In our approach, we partition not only $\mathcal{C}$, but also $\mathcal{P}$, and execute the uncertain dynamic skyline query only for $c\in URS(q)$, not for each $c\in\mathcal{C}$ as suggested in Lemma \ref{lem:influencescorenew}. Our approach is described below.

\begin{algorithm}[t]
\scriptsize
\DontPrintSemicolon
\SetKwInOut{Input}{Input}\SetKwInOut{Output}{Output}
\Input{$q$: query, $\mathcal{P}$: products, $\mathcal{C}$: customers}
\Output{$\tau(q)$: influence score of $q$}
\Begin{

$URS(q)\leftarrow$parallelURS($q$);\tcp*[r]{\tiny{Algorithm \ref{algo:parallelURS}}}
\tcp*[l]{\tiny{partitioned data $\mathcal{P}_j$ is already sent to workers as part of Algorithm \ref{algo:parallelURS}}}

$\mathcal{T}_0:$ \textbf{parallel}\text{ }\For{each $\mathcal{P}_j$}
	{
		$\mathcal{T}_j:root^\mathcal{P}_j\leftarrow$constructPRTree($\mathcal{P}_j$);\tcp*[r]{\tiny{local PR-Tree on $\mathcal{P}_j$}}		
	}

$\mathcal{T}_0:$ \textbf{parallel}\text{ }\For{each $c\in URS(q)$}
	{
	 	$\mathcal{T}_0:$ \textbf{parallel}\text{ }\For{each $\mathcal{T}_j$}
		{		
			$\mathcal{T}_j$: localUDS($c$, $root^\mathcal{P}_j$);\\
			$UDS^c_j\leftarrow\mathcal{T}_j$: sendLocalUDSPoints();\tcp*[r]{\tiny{local UDS points}}		
			\tcp*[l]{\tiny{local UDS scan points}}		
			$UDSScan^c_j\leftarrow\mathcal{T}_j$: sendLocalUDSScanPoints();
		}
	}

$\mathcal{T}_0:$ \textbf{parallel}\text{ }\For{each $c\in URS(q)$}
	{
	 	$\mathcal{T}_j\leftarrow$ selectAnyAvailableWorker($\{\mathcal{T}_j\}$);\tcp*[r]{\tiny{free worker}}		
		$\mathcal{T}_j$: globalUDS($c$, $(\cup UDS^c_j)\cup q$, $\cup UDSScan^c_j$);\\
		$UDS^c\leftarrow\mathcal{T}_j$: sendGlobalUDSPoints();\\
		$UDSScan^c\leftarrow\mathcal{T}_j$: sendGlobalUDSScanPoints();\\		
	}

$\mathcal{T}_0:$ \textbf{parallel}\text{ }\For{each $c\in URS(q)$}
	{
	 	$\mathcal{T}_0:$ \textbf{parallel}\text{ }\For{each $\mathcal{T}_j$}
		{
			sendUDSScanPoints( $\mathcal{T}_j$, $UDSScan^c$);\\
			$\mathcal{T}_j$:dominatingPointSet($UDSScan^c$, $root^\mathcal{P}_j$);\\
			$UDSScanDom^c_j\leftarrow\mathcal{T}_j$:sendLocalDomPoints();\\			
		}
	}

$\tau(q)\leftarrow 0$;\tcp*[r]{\tiny{initialization}}
\textbf{parallel}\text{ }\For{each $c\in URS(q)$}
	{
		updateDSkyProbs($UDSScan^c$, $\cup UDSScanDom^c_j$);\\
		$Pr^c_{Fav}(q)\leftarrow$computeFavProb($UDS^c$, $UDSScan^c$);\\
		$\tau(q)\leftarrow \tau(q)+Pr^c_{Fav}(q)$; \tcp*[r]{\tiny{Eq. \ref{eq:ursscore}}}
	}

}
\normalsize
\caption{Parallel Influence Score}
\label{algo:parallelInfluenceScore}
\end{algorithm}

Firstly, we compute the uncertain reverse skyline of $q$, i.e., $URS(q)$ by calling Algorithm \ref{algo:parallelURS}. Then, each worker constructs the PR-Tree on $\mathcal{P}_j$ without converting it to midpoints. Then, we compute two sets of products $UDS_j$ and $UDSScan_j$ for each customer $c\in URS(q)$ on each partition $\mathcal{P}_j$ locally by following the same technique described in \cite{ZhouLXZL16}. Once the local $UDS_j$ and $UDSScan_j$ product sets are calculated, we accumulate them into the sets $UDS$ and $UDSScan$ in the master. We move a product $p^\prime$ from $UDS$ to $UDSScan$ iff $\exists p\in UDS$ such that $p\ne p^\prime$ and $p\prec_c p^\prime$. We also update $UDSScan$ by ignoring all $p^\prime\in UDSScan$ iff $\exists p\in UDSScan$ such that $p\ne p^\prime$ and $p\prec^u_c p^\prime$. 

Once the $UDS$ and $UDSScan$ product sets are computed for each $c\in URS(q)$, we update the dynamic skyline probabilities of the $UDSScan$\footnote{The dynamic skyline probability of a $p\in UDS^c$ is $Pr(p)$ i.e., $Pr^c_{DSky(p)}=Pr(p)$, as $\not\exists p^\prime\in\mathcal{P}$ such that $p^\prime\prec_c p$.} product set in parallel. To achieve this, firstly we compute the dominating points for each $p\in UDSScan$ on each partition $\mathcal{P}_j$ by running window/range query for it locally. Once done for each partition, we update the dynamic skyline probabilities of the products $UDSScan$ by their dominating products and compute the favorite probability $Pr^c_{Fav}(q)$ of each $c\in URS(q)$ in the master. Once the favorite probabilities are computed, the influence score $\tau(q)$ of the query product $q$ is computed by following Eq. \ref{eq:ursscore}. The above parallel steps are pseudocoded in Algorithm \ref{algo:parallelInfluenceScore}.  

\begin{lemma}
Algorithm \ref{algo:parallelInfluenceScore} accurately computes the influence score of an arbitrary query product $q$ in parallel.
\label{lem:parallelISaccuracy}
\end{lemma}

\begin{proof}
Here, we prove that we accurately compute UDS and UDSScan product sets for each customer $c\in URS(q)$ in Algorithm \ref{algo:parallelInfluenceScore}. The local $UDS_j$and $UDSScan_j$ product sets are computed by following the same the technique as described in \cite{ZhouLXZL16}. Once these sets are computed locally, we accumulated them in the master for further refinement. The refinement ensures that $UDS$ set includes only non-dominating products for a customer $c\in URS(q)$. Similarly, the $UDSScan$ set includes only products that are not UD-dominated by any other products. Finally, the algorithm computes the dominating products for each product $p\in UDSScan$ w.r.t. $c$ by executing range query on each partition $\mathcal{P}_j$ w.r.t. $c$ and $p$. The discovery of these dominating products in each partition are independent from one partition to another. Therefore, the final UDS and UDSScan (along with the dominating products of each $p\in UDSScan$) product sets are accurate. Hence, the lemma.
\end{proof}


\subsection{Optimization} 
\label{sec:paralleloptimization}
An optimized version of Algorithm \ref{algo:parallelURS} can be achieved by applying Lemma \ref{lem:UMSLnodetonodepruning} and Lemma \ref{lem:URSnodecustomersaddition} while computing the local $UMSL$ and $URS$ of $q$, respectively, as we apply these lemmas in Algorithm \ref{algo:optimizedSerialURS}. An optimized version of Algorithm \ref{algo:parallelInfluenceScore} can also be achieved by executing optimized version of Algorithm \ref{algo:parallelURS} while computing the $URS$ of $q$ in line 2.

\section{Experiments}
\label{sec:experiments}

This section compares the efficiencies of different approaches for evaluating the uncertain reverse skyline queries and computing the influence score of a product in probabilistic databases.

\subsection{Datasets, Queries and Environment}

\textbf{Datasets:} We evaluate the efficiency of our pruning ideas and techniques for processing the uncertain reverse skyline queries using real CarDB\footnote{https://autos.yahoo.com/} data which consists of $2\times10^5$ car objects. The CarDB is a six-dimensional dataset with attributes: \emph{make}, \emph{model}, \emph{year}, \emph{price}, \emph{mileage} and \emph{location}. We consider only the three numerical attributes \emph{year}, \emph{price} and \emph{mileage} in our experiments after normalizing them into the range $[0,1]$. We randomly select half of the car objects as products and the rest as the customer preferences. We also assign random probabilities to the car objects. The synthetic data experiments include data: uniform (UN), correlated (CO) and anti-correlated (AC), consisting of varying number of products, customers and dimensions. The cardinalities of the synthetic datasets range from $2$K to $10$M. The dimensionality ($d$) of the datasets varies from 2 to 6. 

\textbf{Test Queries:} The test queries are generated (synthetic) and selected (CarDB) randomly by following the distribution of the respective datasets. Again, the query products are assigned with random probabilities.

\begin{table}[tb]
\vspace{-2ex}
\centering
\setlength{\tabcolsep}{.5pt}
\tiny
\caption{Settings of parameters}
\begin{tabular}{|l|l|} \hlinewd{1.1pt}
\textbf{Parameter}&\textbf{Values}\\ \hline \hline
Tested Datasets& Real (CarDB), Synthetic (UN, CO, AC)\\ \hline
Data Cardinality&2K, 3K, 4K, 6K, 8K, 10K, 100K, 1M, 3M, 5M, 7M, 10M\\ \hline
Dimensionality&$2$D, $3$D, $4$D, $5$D, $6$D\\ \hline
No. of Threads&$1\sim 15$ (1 thread per processor)\\ \hline
MAX \#entries in R-Tree&20, 30, 40, 50, 60 data objects\\ \hlinewd{1.1pt}
\end{tabular}
\label{tab:parameters}
\normalsize
\vspace{-2ex}
\end{table}

\textbf{Computing Environment:} We develop our algorithms in Java and execute them in Swinburne HPC system \footnote{http://www.astronomy.swin.edu.au/supercomputing/} with 1$\sim$15 processors and maximum 60GB main memory, where the parallel computing environment (master-worker) is simulated with Java multi-threading and LOCK-based synchronization. The above parameters are summarized in Table \ref{tab:parameters}.

\begin{table*}[tb]
\setlength{\tabcolsep}{1.4pt}
\vspace{-2ex}
\tiny
\centering
\caption{Effect of customer cardinality on efficiency of evaluating URS queries by different approaches}
\begin{tabular}{|l|c|c|c|c|c|c|c|c|c|c|c|c|} \hlinewd{0.7pt}
\multirow{2}{*}{\textbf{Cardinality}}&\multicolumn{3}{c|}{\textbf{CarDB (millisecs)}}&\multicolumn{3}{c|}{\textbf{UN (millisecs)}}&\multicolumn{3}{c|}{\textbf{CO (millisecs)}}&\multicolumn{3}{c|}{\textbf{AC (millisecs)}}\\ \cline{2-13}
{}&\textbf{SER-URS}&\textbf{OPT-URS}&\textbf{Na\"ive-URS}&\textbf{SER-URS}&\textbf{OPT-URS}&\textbf{Na\"ive-URS}&\textbf{SER-URS}&\textbf{OPT-URS}&\textbf{Na\"ive-URS}&\textbf{SER-URS}&\textbf{OPT-URS}&\textbf{Na\"ive-URS}\\ \hline
 \hline
Customer(2K)& 		3017& 2990&		143803&	2927&	2991&		140145&	3684&	2940&		118851&	3402&	3246&	139054\\ \hline

Customer(4K)&		3067& 3123&		281937&	3084&	3029&		251026&	3251&	3046&		238909&	3399&	3672&	259967\\ \hline

Customer(6K)&		3162& 3136&		419895&	3233&	3355&		380060&	3166&	2913&		337296&	3402&	3679&	356604\\ \hline

Customer(8K)&		3302& 3288&		597125&	3186&	3278&		524370&	3109&	3106&		457902&	3443&	3696&	465955\\ \hline

Customer(10K)&		3303& 3246&		749371&	3468&	3257&		617057&	3230&	3222&		545728&	3837&	4100&	578158\\ \hline

Customer(100K)& 	5077& 5196&	not executed& 	4510&	4756&	not executed&	4657&	5134&	not executed&	5201&	5167&	not executed\\ \hlinewd{0.7pt}
\end{tabular}
\normalsize
\label{tab:URSEffectOfCustomerCardinalities}
\vspace{-2ex}
\end{table*}

\begin{table*}[tb]
\setlength{\tabcolsep}{4pt}
\tiny
\centering
\caption{Effect of customer cardinality on efficiency of computing influence scores by different approaches}
\begin{tabular}{|l|c|c|c|c|c|c|c|c|c|c|c|c|} \hlinewd{0.7pt}
\multirow{2}{*}{\textbf{Cardinality}}&\multicolumn{3}{c|}{\textbf{CarDB (millisecs)}}&\multicolumn{3}{c|}{\textbf{UN (millisecs)}}&\multicolumn{3}{c|}{\textbf{CO (millisecs)}}&\multicolumn{3}{c|}{\textbf{AC (millisecs)}}\\ \cline{2-13}
{}&\textbf{SER-IS}&\textbf{OPT-IS}&\textbf{Na\"ive-IS}&\textbf{SER-IS}&\textbf{OPT-IS}&\textbf{Na\"ive-IS}&\textbf{SER-IS}&\textbf{OPT-IS}&\textbf{Na\"ive-IS}&\textbf{SER-IS}&\textbf{OPT-IS}&\textbf{Na\"ive-IS}\\ \hline
 \hline
Customer(2K)& 		5144& 	 5149&		1350344&	2909&	2907&		  550090&	2797&	2815&		473691&	2980&	2829&	507864\\ \hline

Customer(4K)&		8438& 	 8472&		2636079&	3067&	2962&		1288985&	2872&	2888&		988031&	3091&	2978&	1005732\\ \hline

Customer(6K)&		11748&  11516&		3915923&	6051&	6011&		1609840&	2958&	2920&		1536300&	3045&	3015&	1440399\\ \hline

Customer(8K)&		11953&  11923&		5671686&	6075&	5998&		2135613&	2974&	2911&		2109738&	3111&	3207&	1915065\\ \hline

Customer(10K)&		12262&  12054&		5143220&	5969&	5930&		3027367&	2976&	3116&		2668434&	3172&	3157&	2273715\\ \hline

Customer(100K)& 	13578&  14116&	not executed& 	10595&	11701&	not executed&	9838&	10173&	not executed&	9311&	8430&	not executed\\ \hlinewd{0.7pt}
\end{tabular}
\normalsize
\label{tab:ISEffectOfCustomerCardinalities}
\vspace{-2ex}
\end{table*}

\subsection{Tested Algorithms} 

To compare the efficiency of evaluating uncertain reverse skyline queries, we tested the following algorithms: Serial URS (SER-URS) - Algorithm \ref{algo:serialURS}, Optimized URS (OPT-URS) - Algorithm \ref{algo:optimizedSerialURS}, Parallel URS (PAR-URS) - Algorithm \ref{algo:parallelURS} and Optimized Parallel URS (PAR-URS$^*$) - Optimized Algorithm \ref{algo:parallelURS}. The na\"{\i}ve algorithm proposed in \cite{ZhouLXZL16} and its parallel version are called Na\"{\i}ve-URS and  Na\"{\i}ve-PAR-URS, respectively. To improve the performance of Na\"{\i}ve-URS and Na\"{\i}ve-PAR-URS, we do not update the dynamic skyline probabilities of the products that appear in the UDSScan set of each customer $c\in\mathcal{C}$ as we do not need to know the dynamic skyline probabilities of these products for the inclusion of the customer $c$ in $URS(q)$, we only need to know whether $q$ appears in the UDS or UDSScan sets of $c$. 

To compare the efficiency of computing the influence score of a probabilistic product, we tested the efficiencies of the following algorithms: Serial Influence Score (SER-IS) - Algorithm \ref{algo:serialInfluenceScore}, Optimized Influence Score (OPT-IS) - Optimized Algorithm \ref{algo:serialInfluenceScore}, Parallel Influence Score (PAR-IS) - Algorithm \ref{algo:parallelInfluenceScore} and Optimized Parallel Influence Score (PAR-IS$^*$) - Optimized Algorithm \ref{algo:parallelInfluenceScore}. The na\"{\i}ve algorithm \cite{ZhouLXZL16} and its parallel version are called Na\"{\i}ve-IS and Na\"{\i}ve-PAR-IS, respectively.

\subsection{Efficiency Study}
\label{sec:efficiency}

This section studies the efficiency of our proposed algorithms by comparing the execution times with the na\"{\i}ve approach proposed in \cite{ZhouLXZL16} from the following perspectives.

\subsubsection{Effect of data cardinalities}
\label{sec:EffectOfDataCardinalities}
Here, we examine the effect of data cardinality (\#customers) on the efficiency of processing uncertain reverse skyline queries and computing influence score of a probabilistic product by different approaches on the tested datasets. We set $|\mathcal{P}|$ = 100K, $d$ = 2 and vary $|\mathcal{C}|$ from 2K to 100K. We also set MAX \#entries in a R-Tree node to 50. We run a number of queries and the results of evaluating a uncertain reverse skyline query and computing the influence score of a probabilistic product on average are shown in Table \ref{tab:URSEffectOfCustomerCardinalities} and Table \ref{tab:ISEffectOfCustomerCardinalities}, respectively. It is evident that the na\"{\i}ve approach \cite{ZhouLXZL16} is not scalable, whereas our approaches are scalable and can finish their executions within seconds even for 100K customers (na\"{\i}ve approach \cite{ZhouLXZL16} is not executed as it takes hours to finish). We see that the speed-ups achieved by our approach over the na\"{\i}ve approach \cite{ZhouLXZL16} are hugely significant. 

To justify the scalability of our approaches for millions of data objects, we perform another two experiments in UN dataset. For the first experiment, we set $|\mathcal{C}|=1$M and vary $|\mathcal{P}|$ from $1$M to $10$M. For the second experiment, we set $|\mathcal{P}|=1$M and vary $|\mathcal{C}|$ from $1$M to $10$M. For both experiments, we also set $d$ = 2 and MAX \#entries in a R-Tree node to 50. Finally, we run a number of queries and the results of evaluating a uncertain reverse skyline query and computing the influence score of a probabilistic product on average are shown in Fig. \ref{fig:UN2DURSEffectOfDataCardinality} and Fig. \ref{fig:UN2DISEffectOfDataCardinality}, respectively. We observe that our approaches can finish their executions within few minutes for millions of data objects.

\begin{figure}[t]
\small
\centering
\begin{tabular}{cc}
\begin{minipage}[a]{0.5\linewidth}
\subfigure[]{
\includegraphics[scale=0.32]{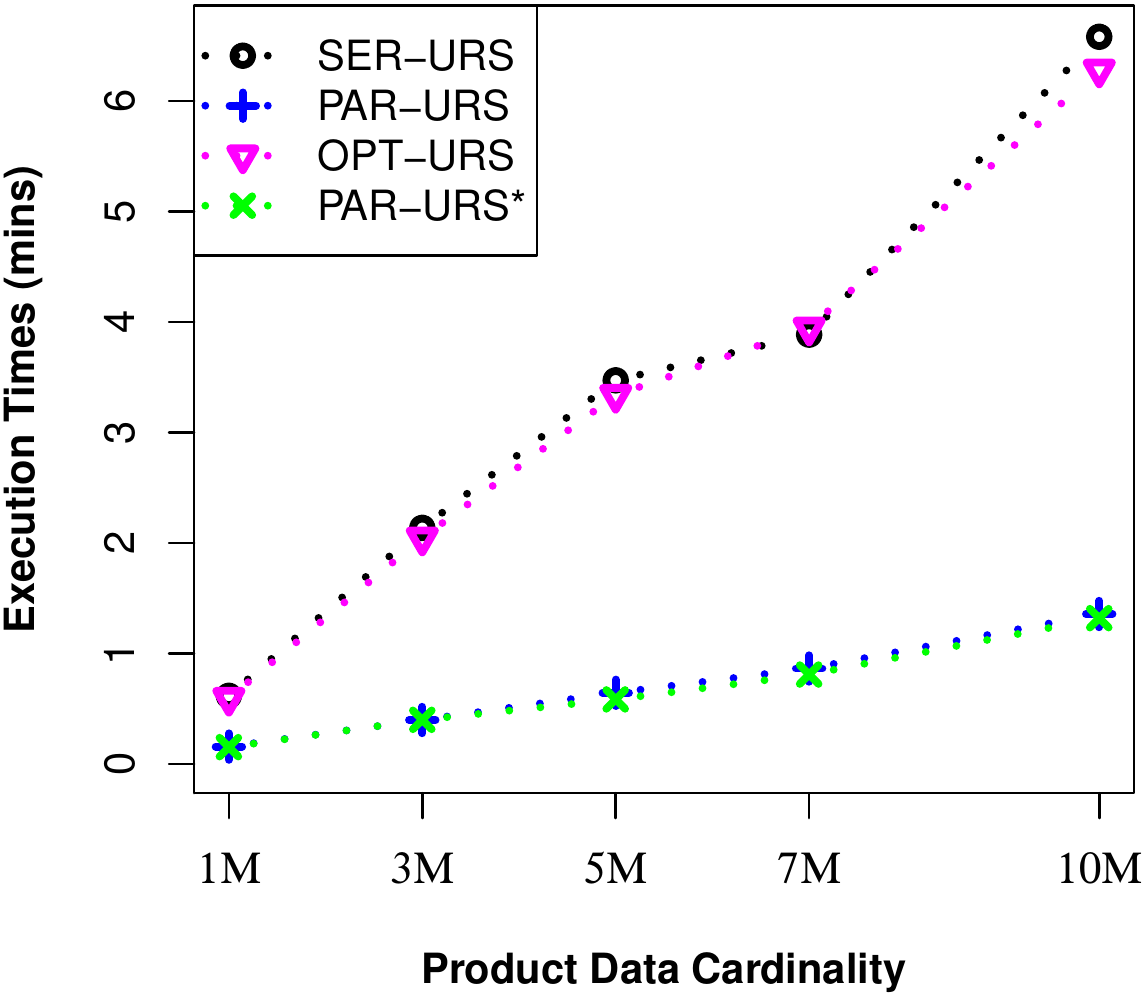}
}
\end{minipage}
&
\begin{minipage}[a]{0.5\linewidth}
\subfigure[]{
\includegraphics[scale=0.32]{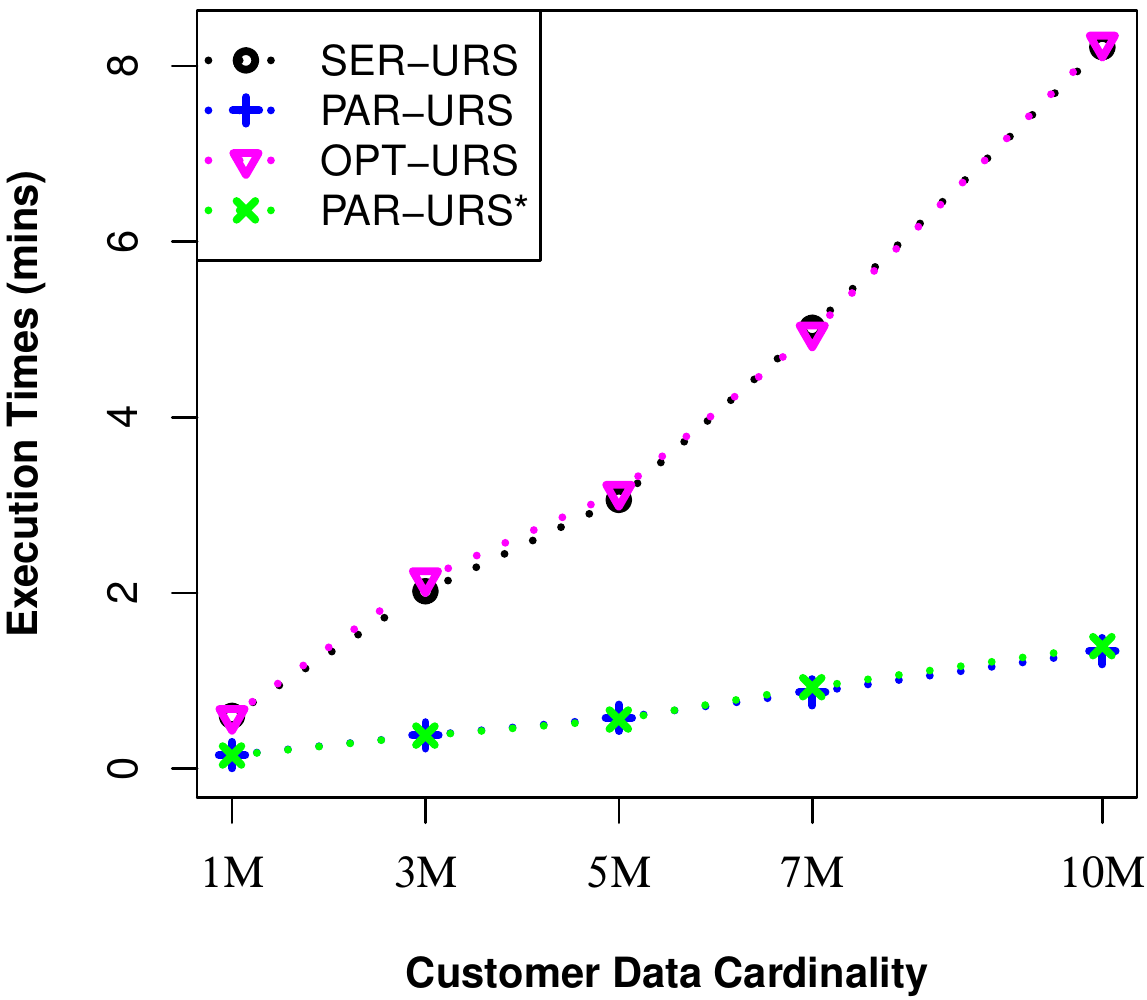}
}
\end{minipage}
\end{tabular}
\vspace{-2ex}
\caption{Effect of data cardinality on the efficiency of processing URS queries in UN dataset: (a) Product Cardinality and (b) Customer Cardinality}
\label{fig:UN2DURSEffectOfDataCardinality}
\vspace{-2ex}
\normalsize
\end{figure}

\begin{figure}[t]
\small
\centering
\begin{tabular}{cc}
\begin{minipage}[a]{0.5\linewidth}
\subfigure[]{
\includegraphics[scale=0.32]{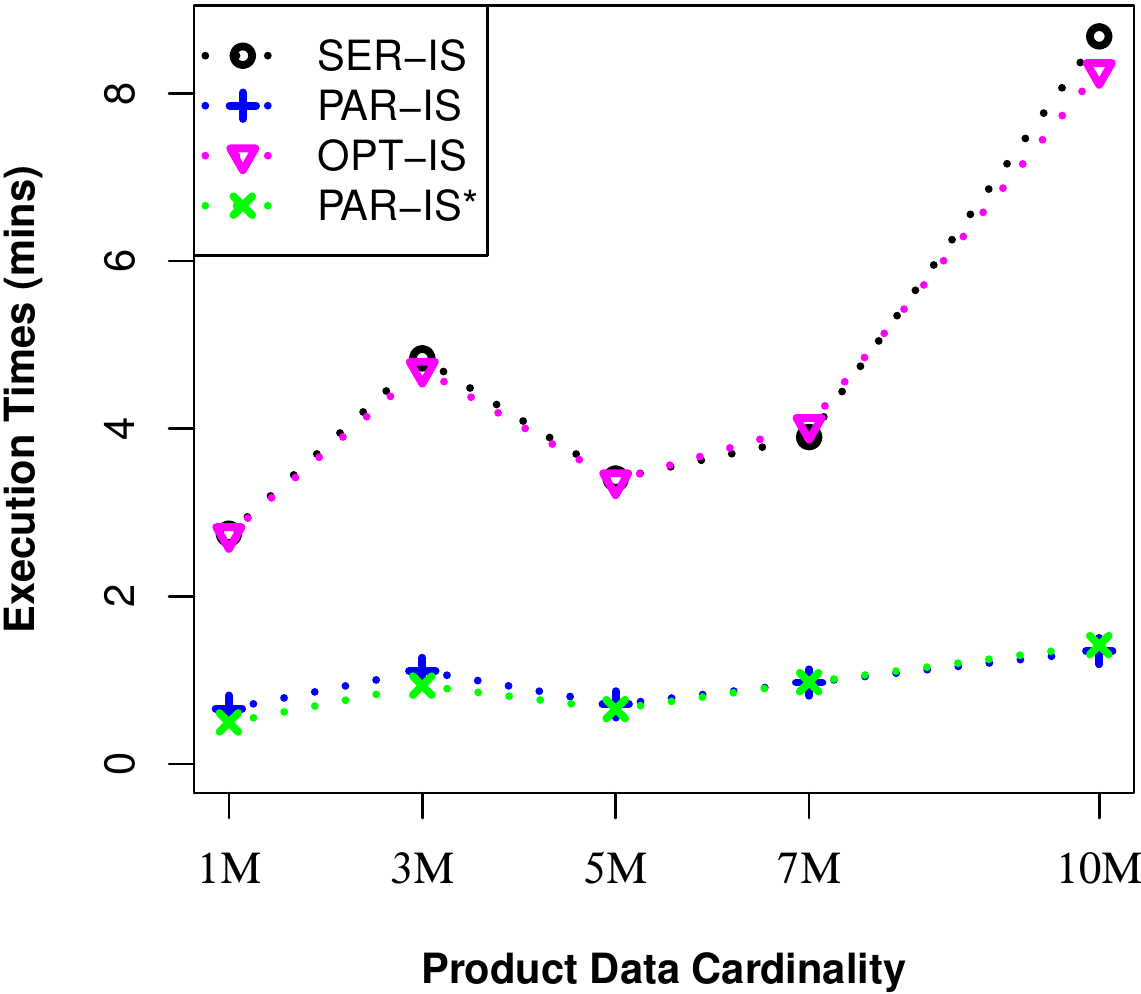}
}
\end{minipage}
&
\begin{minipage}[a]{0.5\linewidth}
\subfigure[]{
\includegraphics[scale=0.32]{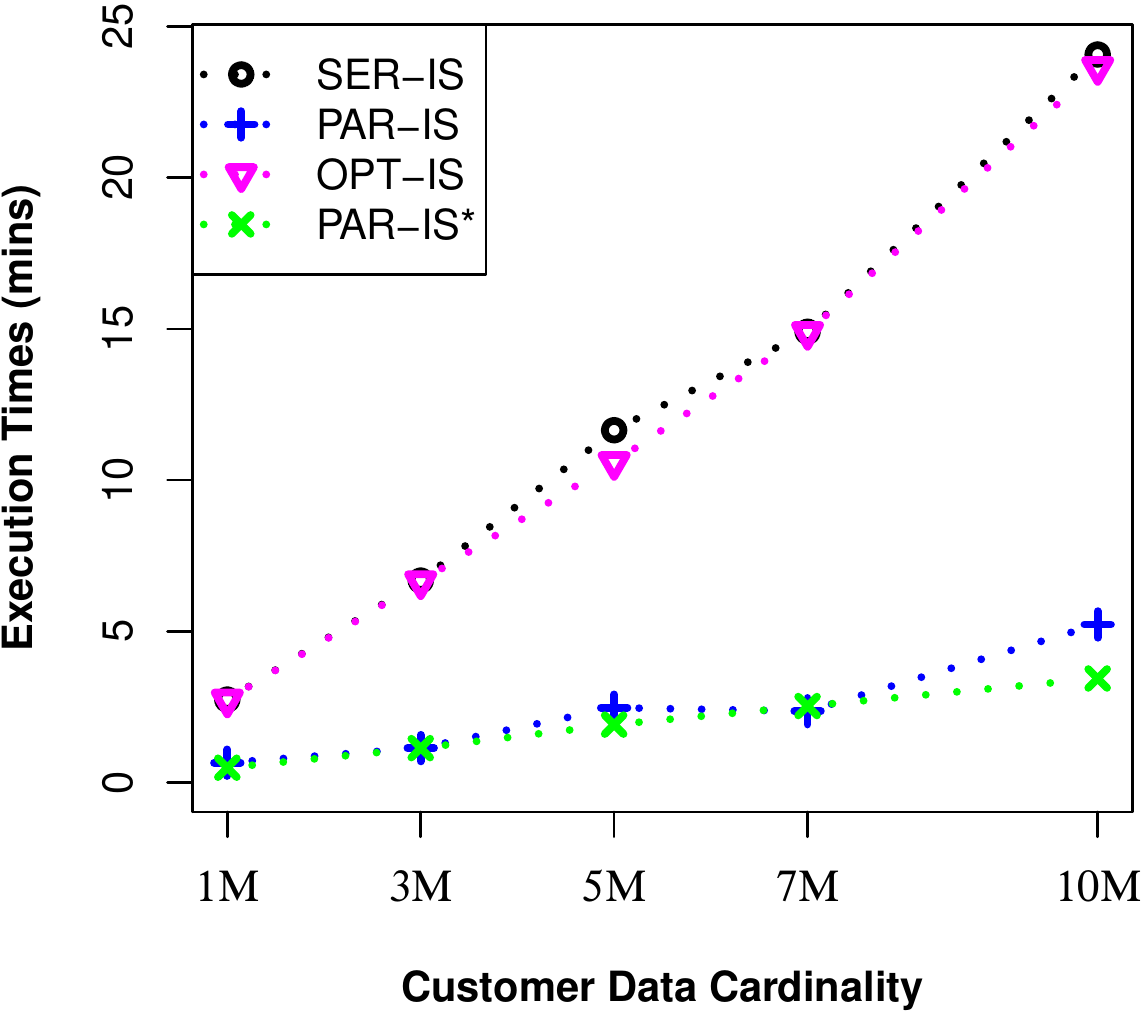}
}
\end{minipage}
\end{tabular}
\vspace{-2ex}
\caption{Effect of data cardinality on the efficiency of computing influence scores in UN dataset: (a) Product Cardinality and (b) Customer Cardinality}
\label{fig:UN2DISEffectOfDataCardinality}
\vspace{-2ex}
\normalsize
\end{figure}

\begin{figure}[tb]
\centering
\includegraphics[scale=0.35]{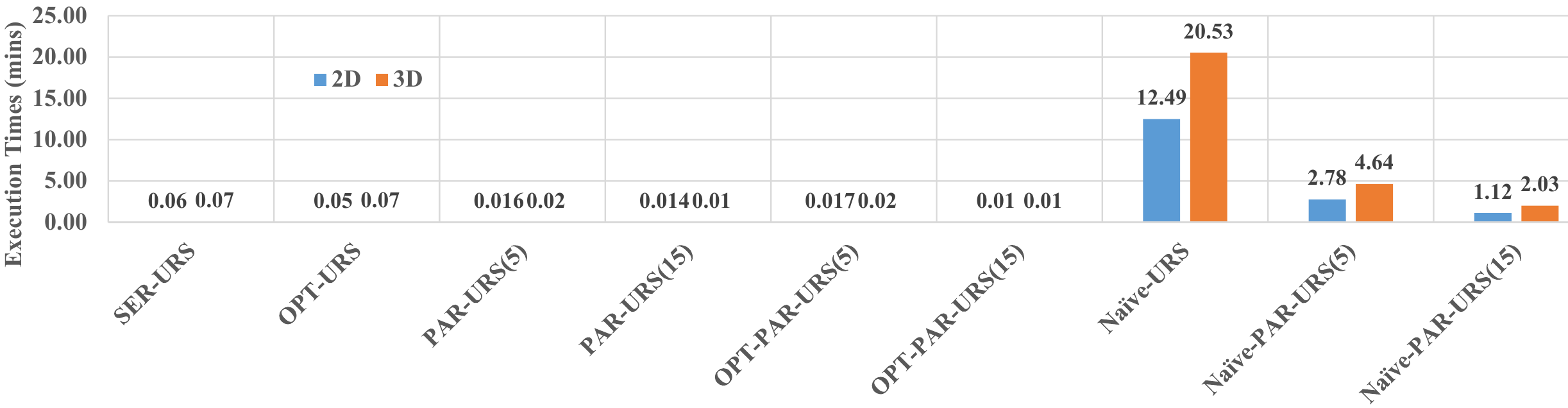}
\vspace{-4ex}
\caption{Effect of dimensions on efficiency of processing reverse skyline queries in CarDB dataset}
\label{fig:CarDBURSEffectOfDimensions}
\vspace{-2ex}
\normalsize
\end{figure}

\begin{figure}[tb]
\includegraphics[scale=0.36]{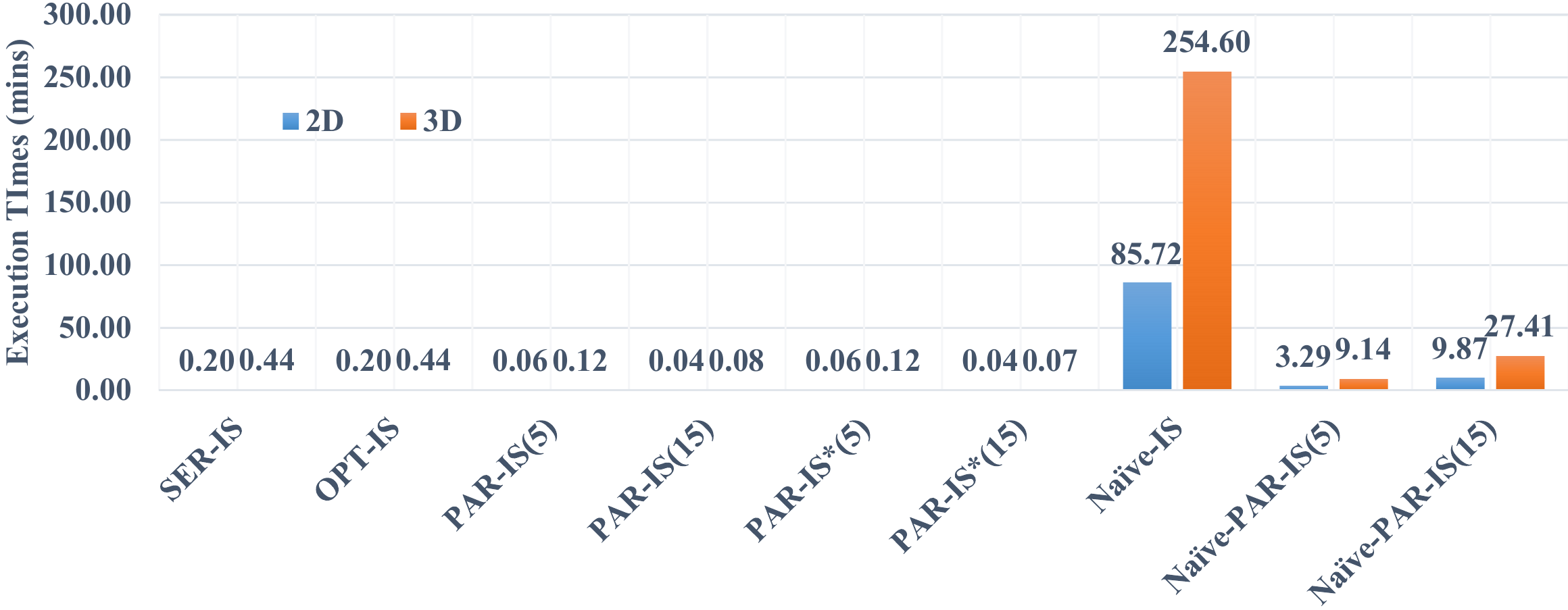}
\vspace{-4ex}
\caption{Effect of dimensions on the efficiency of computing influence scores in CarDB dataset.}
\label{fig:CarDBISEffectOfDimensions}
\vspace{-2ex}
\normalsize
\end{figure}

\begin{figure}[t]
\centering
\includegraphics[scale=0.30]{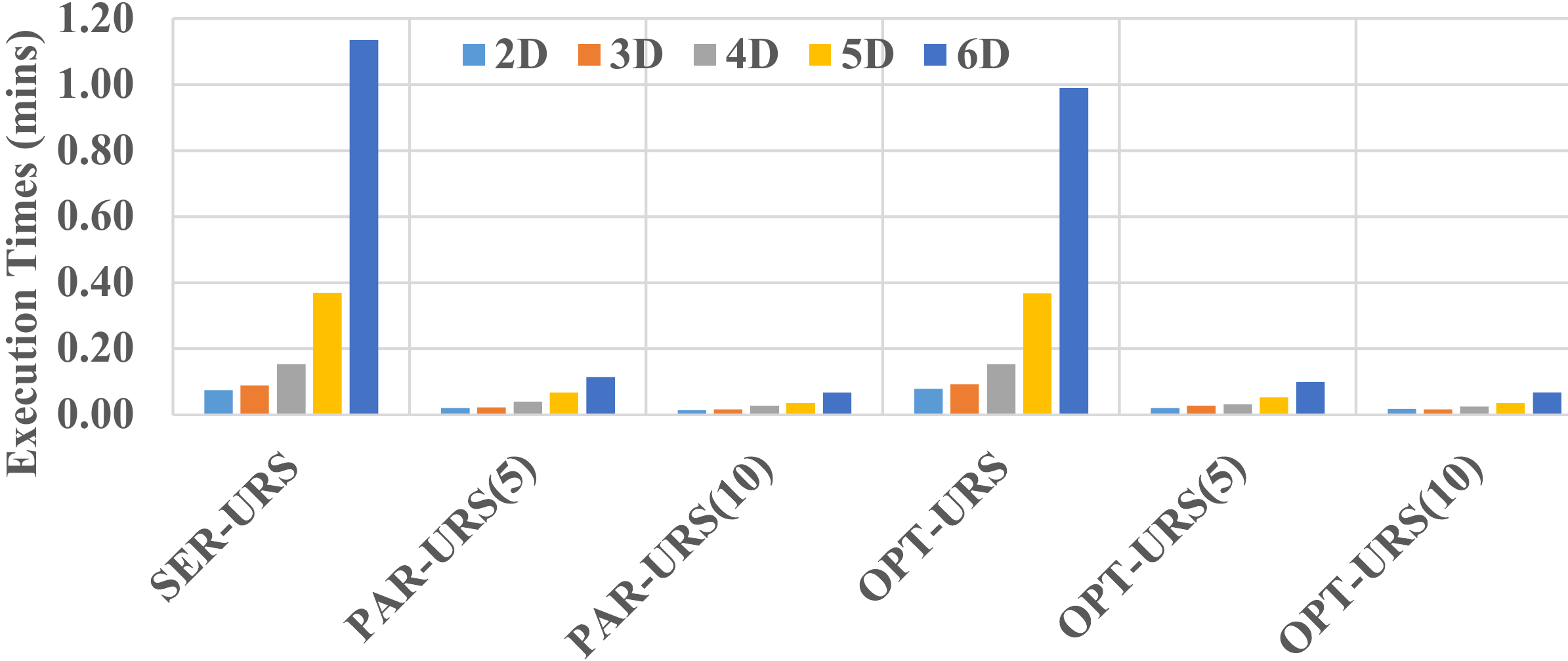}
\vspace{-2ex}
\caption{Effect of dimensions on efficiency of processing reverse skyline queries in UN dataset}
\label{fig:UNURSEffectOfDimensions}
\normalsize
\end{figure}

\subsubsection{Effect of data dimensions}
\label{sec:EffectOfDimensions}
Here, we examine the effect of data dimensionality on the efficiency of processing uncertain reverse skyline queries and computing the influence scores of probabilistic products by different approaches on CarDB two-dimensional (2D) and three-dimensional (3D) datasets. We set $|\mathcal{P}|$ = 100K, $|\mathcal{C}|$=10K, \#threads to 5 and 15 for PAR-URS, PAR-URS*, Na\"{\i}ve-PAR-URS, PAR-IS, PAR-IS* and Na\"ive-PAR-IS, and the MAX \#entries in a R-Tree node to 50. We run a number of queries and the results of processing uncertain reverse skyline of a query and computing the influence score of a probabilistic product on average are shown in Fig. \ref{fig:CarDBURSEffectOfDimensions} and Fig. \ref{fig:CarDBISEffectOfDimensions}, respectively. We observe that the na\"{\i}ve approach\cite{ZhouLXZL16} takes minutes to finish its execution in 3D data even with 15 threads (processors). The execution times get more worse for increased customer cardinality and dimensionality. On the other hand, all of our proposed approaches scale very well and finish their executions within seconds. We also perform another experiment in higher dimensions for UN dataset with varying $d$ from 2 to 6 for testing the efficiency of evaluating the uncertain reverse skyline of a query. For this experiment, we set $|\mathcal{P}|$ and $|\mathcal{C}|$ to $100$K, and the MAX \#entries in a R-Tree node to 50. The results are shown in Fig. \ref{fig:UNURSEffectOfDimensions}. We observe that all of our approaches can finish their executions within 2 minutes. Therefore, we claim that our approaches are scalable even in higher dimensions.

\begin{figure}[t]
\small
\centering
\begin{tabular}{cc}
\begin{minipage}[a]{0.5\linewidth}
\subfigure[CarDB]{
\includegraphics[scale=0.32]{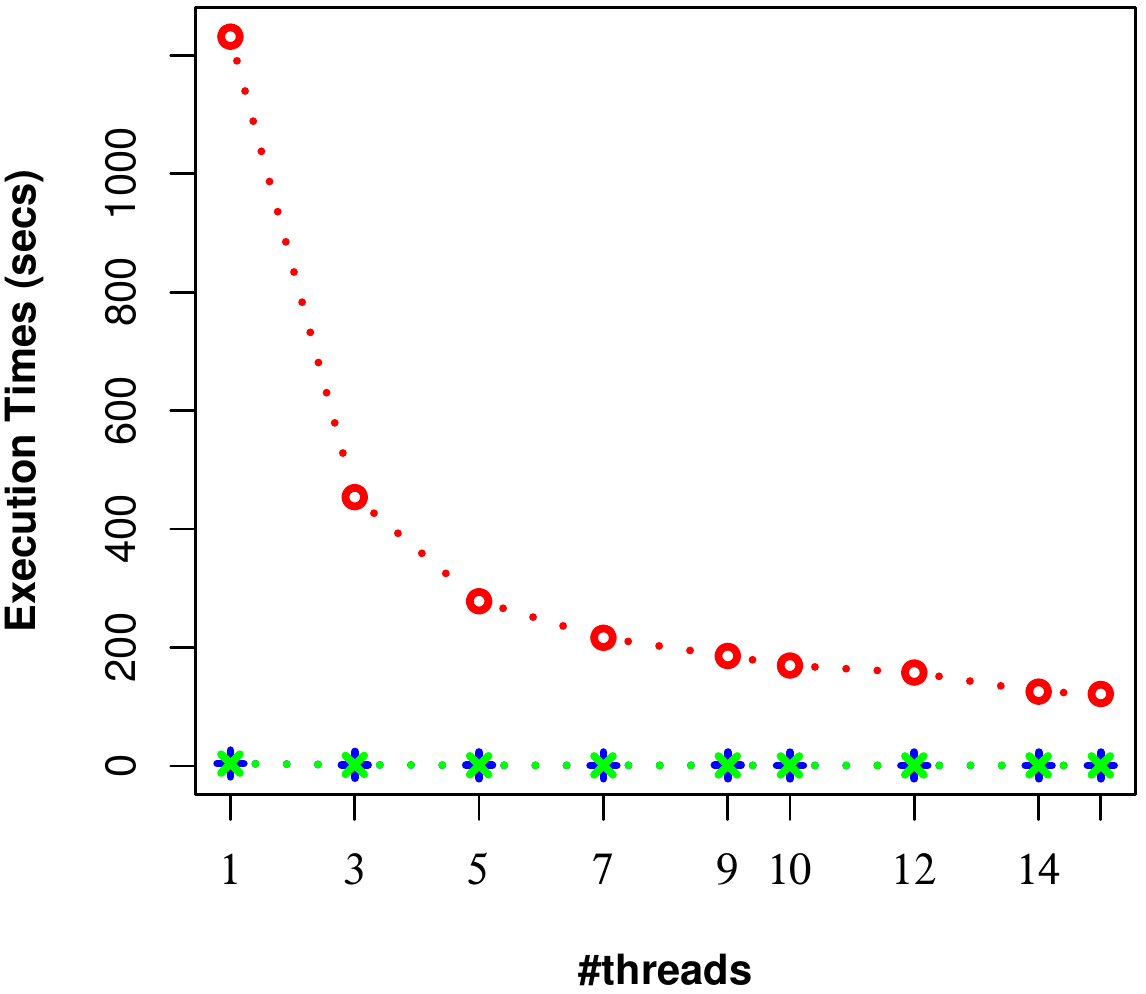}
}
\end{minipage}
&
\begin{minipage}[a]{0.5\linewidth}
\subfigure[UN]{
\includegraphics[scale=0.32]{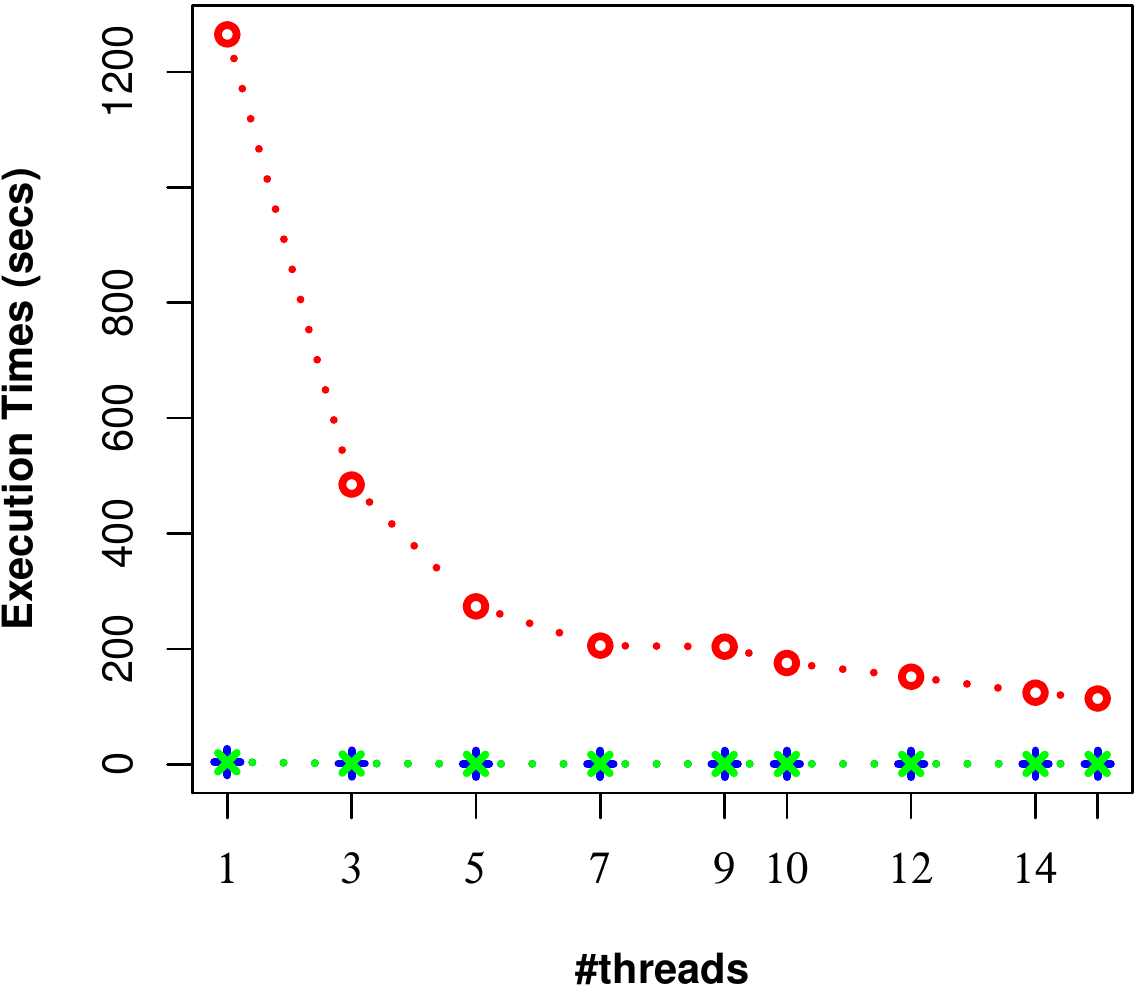}
}
\end{minipage}
\end{tabular}
\vspace{-3ex}
\caption{Effect of \#threads on efficiency of URS queries: (a) CarDB and (b) UN datasets}
\label{fig:3DURSEffectOfThreads}
\vspace{-2ex}
\normalsize
\end{figure}

\begin{figure}[t]
\small
\centering
\begin{tabular}{cc}
\begin{minipage}[a]{0.5\linewidth}
\subfigure[CarDB]{
\includegraphics[scale=0.32]{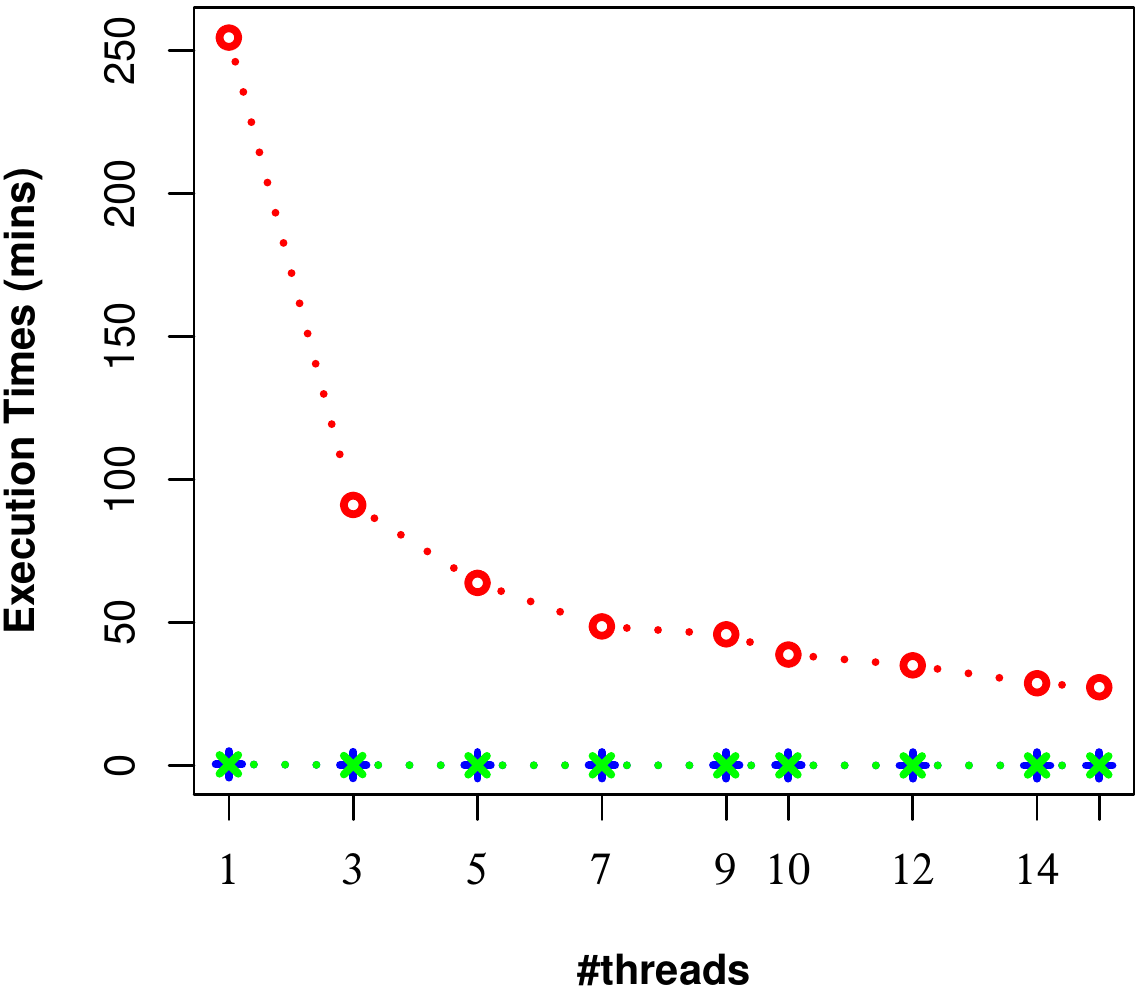}
}
\end{minipage}
&
\begin{minipage}[a]{0.5\linewidth}
\subfigure[UN]{
\includegraphics[scale=0.32]{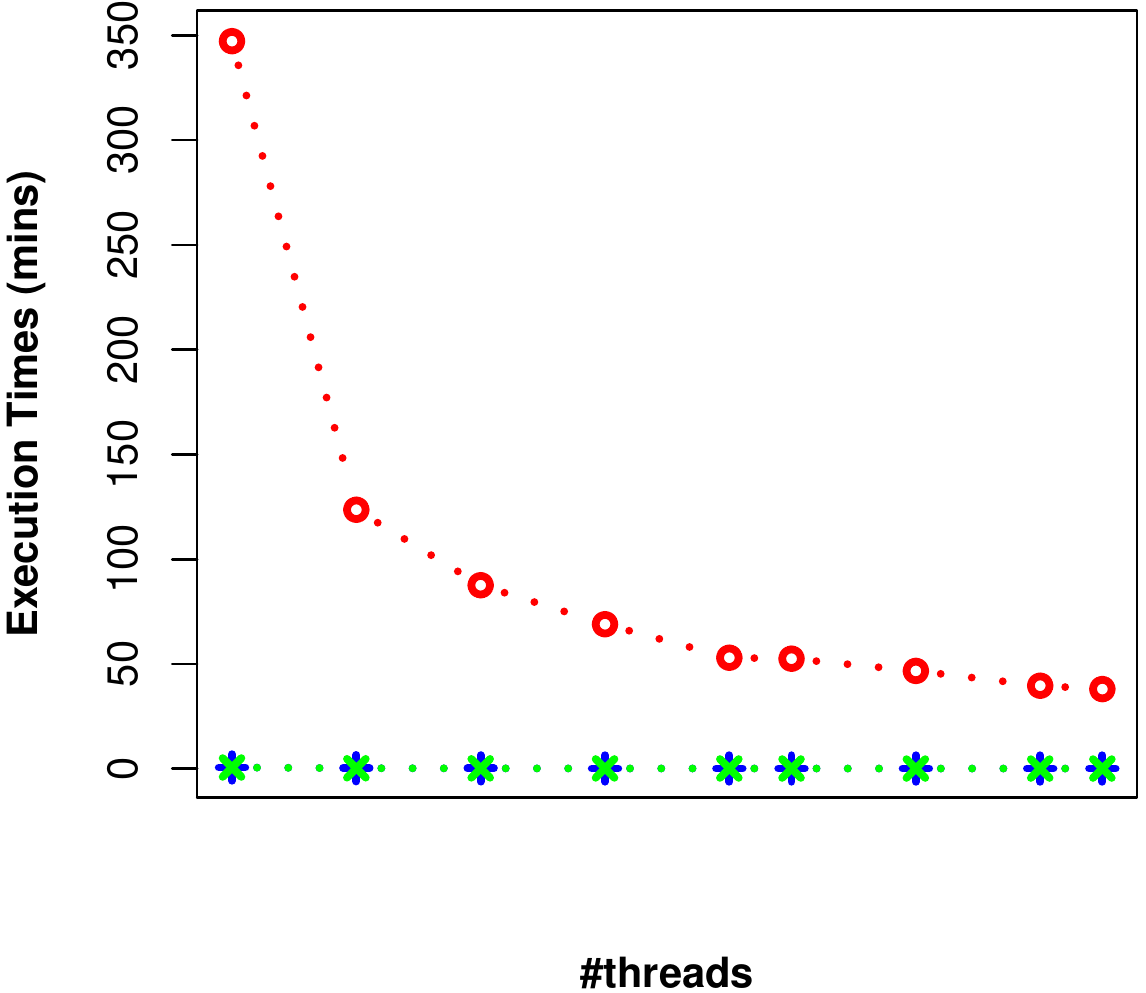}
}
\end{minipage}
\end{tabular}
\vspace{-3ex}
\caption{Effect of \#threads on efficiency of influence scores in: (a) CarDB and (b) UN datasets}
\label{fig:3DISEffectOfThreads}
\vspace{-2ex}
\normalsize
\end{figure}

\subsubsection{Effect of threads}
\label{sec:EffectOfThreads}
Here, we examine the effect of \#threads on the efficiency of processing uncertain reverse skyline queries and computing the  influence scores of probabilistic products in parallel by different approaches on CarDB and UN datasets. We set $|\mathcal{P}|$ = 100K, $|\mathcal{C}|$=10K, $d$ = 3, MAX \#entries in a R-Tree node to 50 and vary \#threads  from 1 to 15. We run a number of queries and the results of evaluating an uncertain reverse skyline query and computing the influence score of a probabilistic product on average for different \#threads are shown in Fig. \ref{fig:3DURSEffectOfThreads} and Fig. \ref{fig:3DISEffectOfThreads}, respectively. It is evident that the na\"{\i}ve approach\cite{ZhouLXZL16} is not scalable even if we increase the \#threads, whereas our approaches are scalable and can finish their executions within seconds with less \#threads. 

\subsubsection{Effect of R-Tree parameters}
\label{sec:EffectOfRTreeParams}
Here, we examine the effect of R-Tree parameters (MAX \#entries in a R-Tree node) on the efficiency of processing uncertain reverse skyline queries and computing the influence scores of probabilistic products by different approaches on CarDB and AC datasets. Here, we set $|\mathcal{P}|$ = 100K, $|\mathcal{C}|$=100K, \#threads to 10 for PAR-URS and PAR-URS*, $d$ = 2 and vary MAX \#entries in a R-Tree node from 20 to 60. We run a number of queries and the results of evaluating an uncertain reverse skyline query and computing the influence score of a probabilistic product on average are shown in Fig. \ref{fig:2DURSEffectOfRTreeParams} and Fig. \ref{fig:2DISEffectOfRTreeParams}, respectively. We observe that efficiency improves in general in SER-URS and OPT-URS with the increased MAX \#entries in a R-Tree node. However, we observe an exception in their parallel evaluations. We also observe that the efficiencies of different approaches improve if we increase the MAX \#entries in a R-Tree node in general except for SER-IS in AC dataset. We believe that the efficiency depends on many factors including data distribution in different threads (processors) and \#threads, not only on the MAX \#entries in a R-Tree node.

\begin{figure}[tb]
\small
\centering
\begin{tabular}{cc}
\begin{minipage}[a]{0.5\linewidth}
\subfigure[CarDB]{
\includegraphics[scale=0.32]{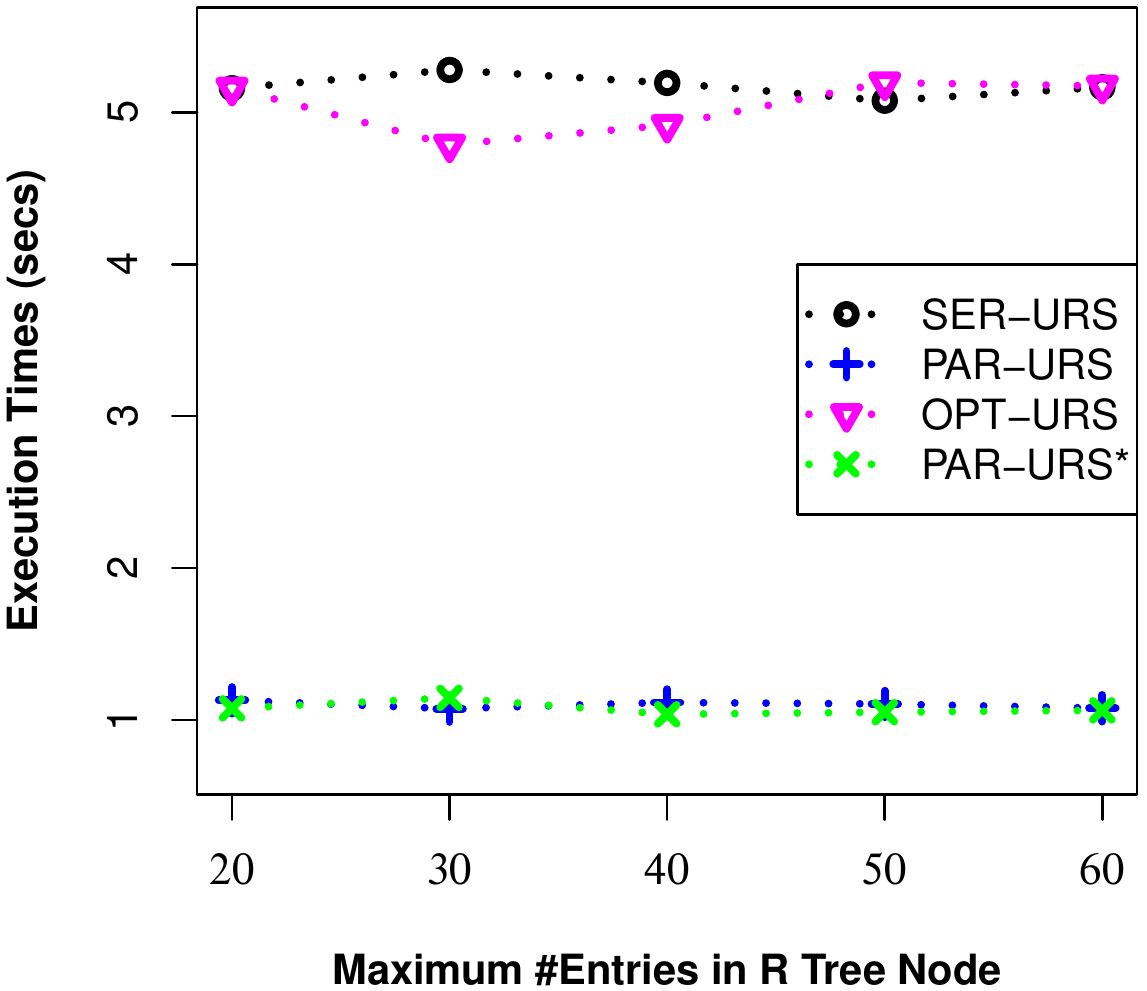}
}
\end{minipage}
&
\begin{minipage}[a]{0.5\linewidth}
\subfigure[AC]{
\includegraphics[scale=0.32]{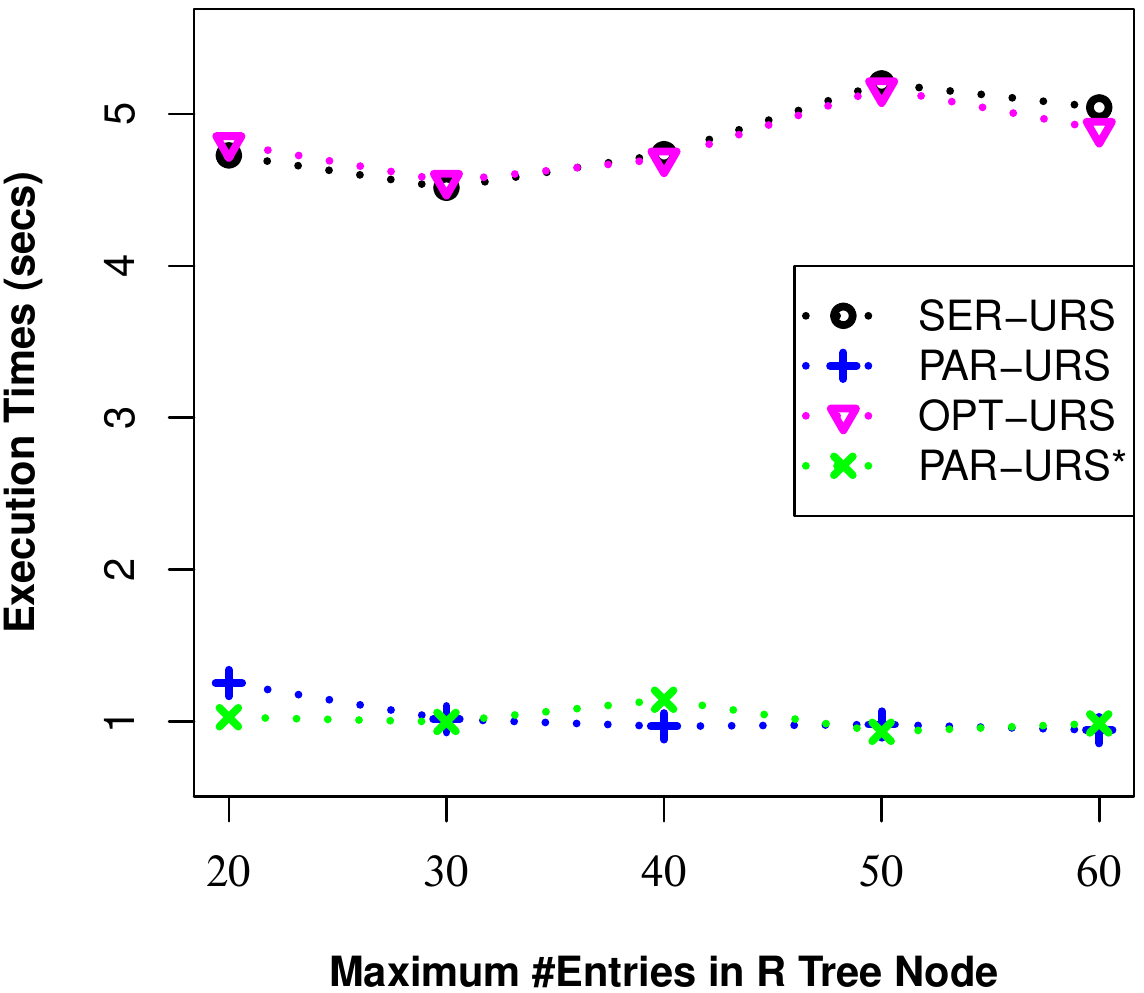}
}
\end{minipage}
\end{tabular}
\vspace{-2ex}
\caption{Effect of R-Tree \#entries on efficiency of URS queries: (a) CarDB and (b) AC datasets}
\label{fig:2DURSEffectOfRTreeParams}
\vspace{-2ex}
\normalsize
\end{figure}

\begin{figure}[t]
\small
\centering
\begin{tabular}{cc}
\begin{minipage}[a]{0.5\linewidth}
\subfigure[CarDB]{
\includegraphics[scale=0.32]{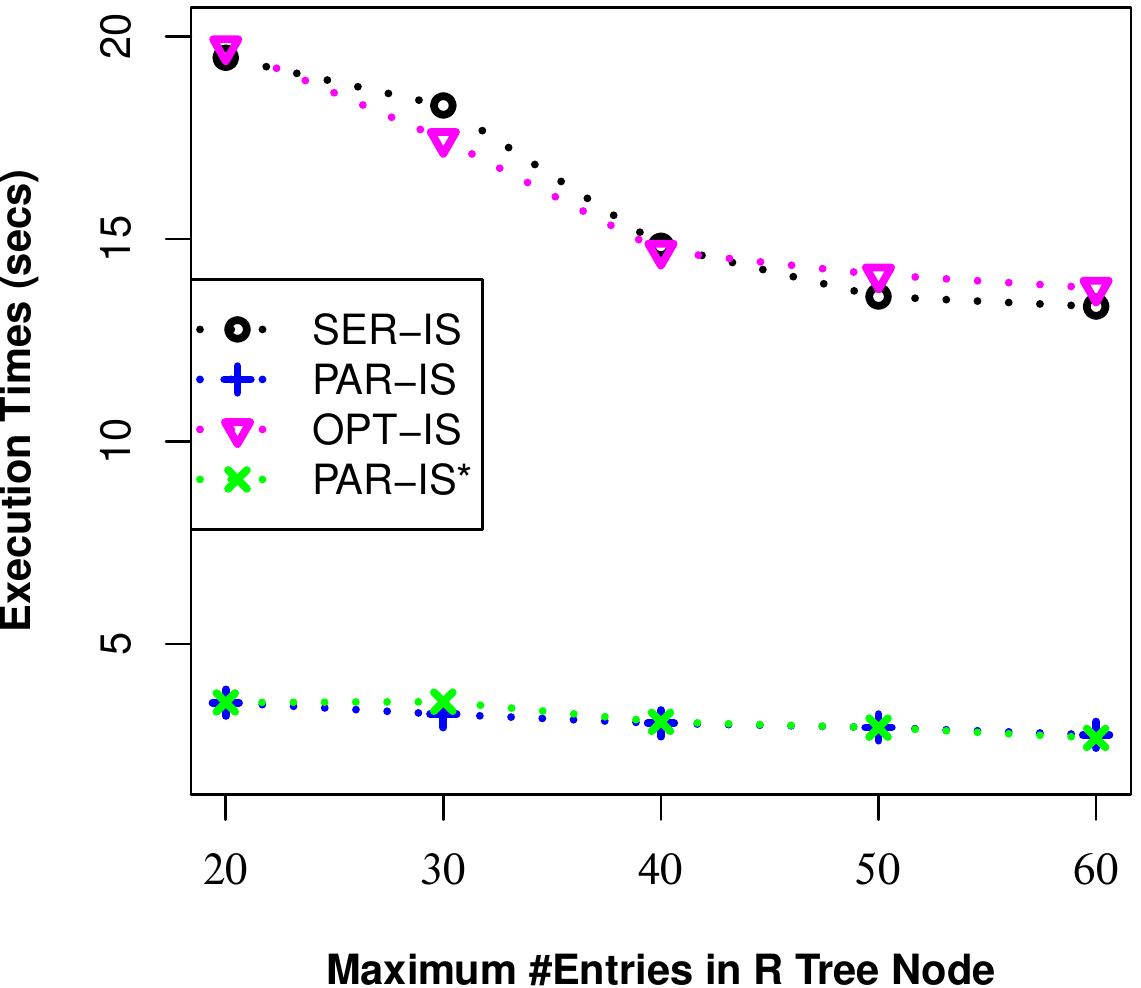}
}
\end{minipage}
&
\begin{minipage}[a]{0.5\linewidth}
\subfigure[AC]{
\includegraphics[scale=0.32]{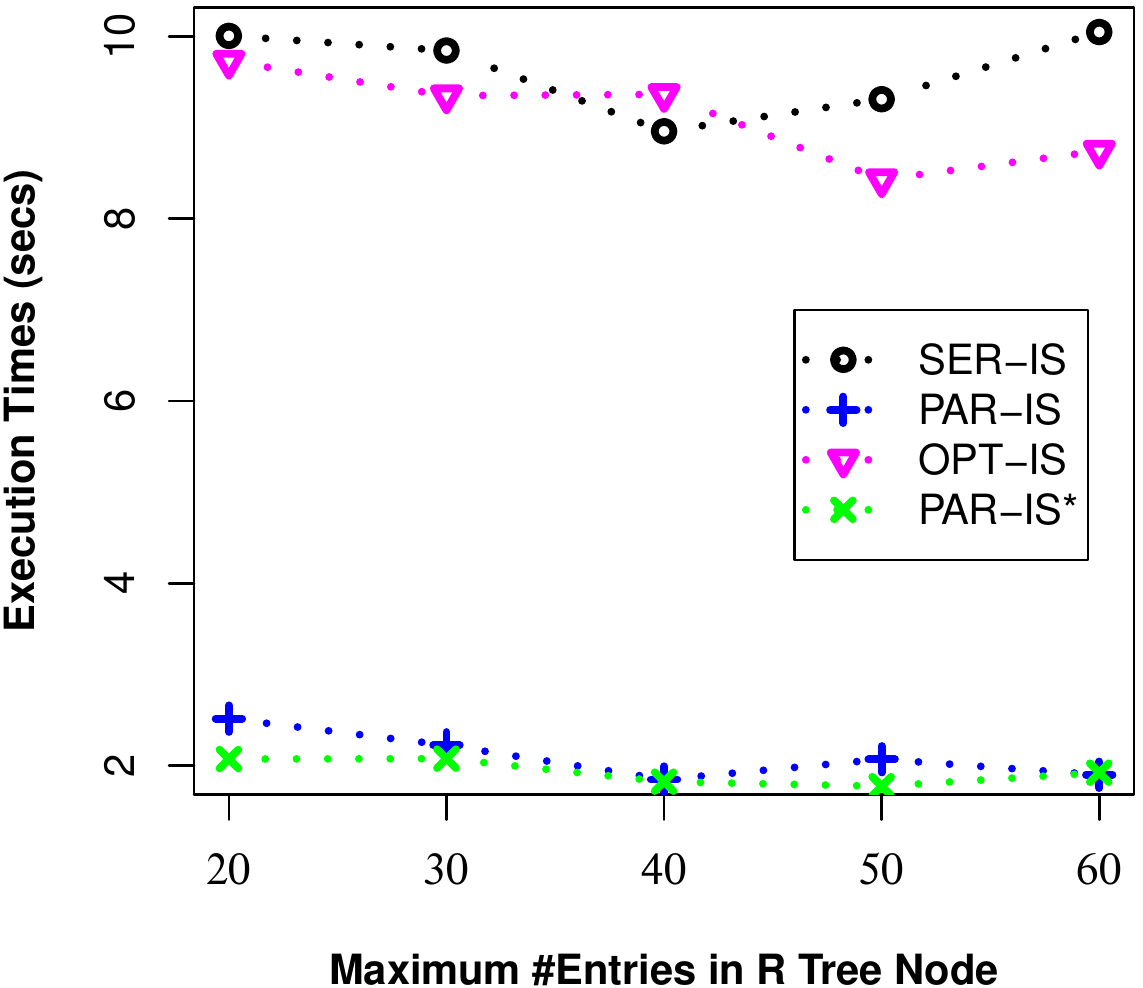}
}
\end{minipage}
\end{tabular}
\vspace{-2ex}
\caption{Effect of R-Tree \#entries on efficiency of influence scores in: (a) CarDB and (b) AC datasets}
\label{fig:2DISEffectOfRTreeParams}
\vspace{-2ex}
\normalsize
\end{figure}

\subsection{Summary}
\label{sec:summary}
We experimentally demonstrate (prove theoretically in Section \ref{sec:complexity}) that the na\"{\i}ve approach proposed in \cite{ZhouLXZL16} is not scalable for computing the influence score of a probabilistic product. The computation of the influence score of a probabilistic product through uncertain reverse skyline in uncertain data is scalable for millions of customer and product data objects, and can finish executions within few minutes.

\section{Related Work}
\label{sec:relatedwork}

\textbf{Reverse Skyline Queries and Related Studies.} Dellis et al. \cite{DellisS07} are the first to present reverse skyline query to the database community. Later, Wu et al. \cite{WuTWDY09} propose an efficient approach for computing the influence of a product through its reverse skyline, where the influence set consists of the member of the reverse skyline query results. Then, \cite{DeshpandeP11} propose an approach for evaluating reverse skyline queries with non-metric similarity measures. Wang et al. \cite{WangXCL12} propose an energy efficient approach for evaluating reverse skyline queries over wireless sensor networks. Arvanitis et al. \cite{ArvanitisDV12} extends this idea for computing the $k$-most attractive candidates ($k$-MAC) from a given set of products that maximizes the size of their joint influence set (score). Islam et al. \cite{IslamLZ13} propose an approach to answer how to turn up a given customer into the reverse skyline query result of an arbitrary query product. Recently, Islam et al. \cite{IslamL16} present an approach for computing the $k$-most promising products ($k$-MPP), which assigns equal probabilities to the products appearing in the dynamic skyline of a customer and selects a subset of given products to maximize their joint probabilistic influence score. All of the above works are in certain data settings. Lian et al. \cite{LianC08}, \cite{LianC13} extends the idea of reverse skyline query in uncertain data settings. However, the probabilistic reverse skylines proposed in \cite{LianC08}, \cite{LianC13} lack friendliness, stability and fairness as per \cite{ZhouLXZL16}.  Zhou et al. \cite{ZhouLXZL16} propose uncertain dynamic skyline and an approach to compute top-$k$ favorite probabilistic products through uncertain dynamic skyline. However, the approach    proposed in \cite{ZhouLXZL16} is not efficient as discussed in Section \ref{sec:complexity}. This paper presents uncertain reverse skyline query to efficiently evaluate the influence of an arbitrary probabilistic product in uncertain data settings. Unlike \cite{LianC08}, \cite{LianC13}, the uncertain reverse skyline proposed here is user friendly, stable and fair.   

\textbf{Parallelizing Reverse Skyline Queries.} Though there exist many works on parallelizing the standard skyline queries (\cite{HoseV12}, \cite{MullesgaardPLZ14}, \cite{AfratiKSU15}, \cite{PertesisD15}, \cite{BoghCA15}, \cite{ZhangJKQ16} for survey), there are only few works devoted to parallelizing the reverse skyline queries. Park et al. \cite{ParkMS13} propose an approach for parallelizing both dynamic and reverse skyline queries in MapReduce by inventing a novel quad-tree based data indexing. Later, the authors extend their quad-tree based data indexing in \cite{ParkMS15} for evaluating probabilistic dynamic and reverse skylines. Recently, Islam et al. \cite{IslamLRA16} propose an advancement of the quad-tree based data indexing proposed in \cite{ParkMS13} for evaluating the dynamic skyline, monochromatic and bichromatic reverse skylines in parallel. Here, we propose an efficient approach for parallelizing the computation of uncertain reverse skyline query result and the influence score of an arbitrary probabilistic product using R-Tree. Our approach for computing the influence score of a probabilistic product is significantly different from the one proposed in \cite{ZhouLXZL16}. Here, we only compute the dynamic skyline probabilities of the products that appear in the uncertain dynamic skyline of the customers existing in the uncertain reverse skyline of the query product, not for all customers in the dataset.  

\section{Conclusion}
\label{sec:conclusion}
This paper presents a novel skyline query, called uncertain reverse skyline, for measuring the influence of an arbitrary probabilistic product in uncertain data settings. We propose efficient pruning ideas and techniques for processing the uncertain reverse skyline and the influence score of a query product in probabilistic databases using R-Tree. We also present a parallel approach for evaluating the uncertain reverse skyline query and the influence score of a probabilistic product, which outperforms its serial counterpart. We conduct experiments with both real and synthetic datasets and compare our results with the existing baseline approach to demonstrate the efficiency of our approach. 

\section{Acknowledgment}
The research of C. Liu and T. Anwar is supported by the ARC discovery projects DP160102412 and DP170104747.

\small
\bibliographystyle{abbrv}
\bibliography{TPDS2017}  
\vspace{-1ex}

\end{document}